\documentclass[a4paper,11pt]{article}
\pdfoutput=1 

\usepackage{jheppub} 
\usepackage[bottom]{footmisc}
\usepackage{amssymb}
\usepackage{amsmath}
\usepackage{amsthm}
\usepackage[usenames,dvipsnames]{xcolor}
\usepackage{epsfig}
\usepackage{dcolumn}
\usepackage{tikz}
\usetikzlibrary{shapes.geometric, arrows}
\usepackage{upgreek}
\usepackage{setspace}
\usepackage{enumitem}
\usepackage{array,multirow,bigdelim,arydshln}
\usepackage{appendix}
\usepackage{xparse}
\usepackage{stmaryrd}
\usepackage[T1]{fontenc} 
\usepackage{mathtools}
\usepackage{physics} 
\usepackage{adjustbox}
\usepackage{multirow}
\usepackage{graphicx} 
\usepackage{float} 
\graphicspath{{./images/}}
\usepackage[nottoc]{tocbibind}
\usepackage{hyperref}
\usepackage[utf8]{inputenc}
\hypersetup{
	colorlinks,
	urlcolor=Maroon,
	linkcolor=Maroon,
	citecolor=Maroon
	}

\NewDocumentCommand{\binomial}{omm}
 {%
  \genfrac(){0pt}{}{#2}{#3}%
  \IfValueT{#1}{_{\!#1}}%
 }
\NewDocumentCommand{\eulerian}{omm}
 {%
  \genfrac<>{0pt}{}{#2}{#3}%
  \IfValueT{#1}{_{\!#1}}%
 }

\def \s {\sigma}

\usepackage{latexsym}
\usepackage{tikz}

\theoremstyle{plain}
\newtheorem{lemma}{Lemma}[section]

\newtheorem{thm}{Theorem}[section]

\newtheorem{defn}[thm]{Definition}
\theoremstyle{definition}

\newtheorem{conj}[thm]{Conjecture}

\usepackage[percent]{overpic}
\usepackage{multirow} 
\usepackage{slashed}

\def\yz#1\yz {{\color{blue} [[YZ: #1]] }}
\def\bu#1\bu {{\color{red} [[BU: #1]] }}

\title{$\phi^p$ Amplitudes from the Positive Tropical Grassmannian: Triangulations of Extended Diagrams}

\author[a,b]{Bruno Gim\'enez Umbert}\emailAdd{b.gimenez-umbert@soton.ac.uk}
\author[c,b]{and Karen Yeats}\emailAdd{kayeats@uwaterloo.ca}

\affiliation[a]{School of Physics \& Astronomy, University of Southampton, Southampton, SO17 1BJ, UK\\ Mathematical Sciences, University of Southampton, Highfield, Southampton SO17 1BJ, UK}
\affiliation[b]{Perimeter Institute for Theoretical Physics, Waterloo, ON N2L 2Y5, Canada}
\affiliation[c]{Department of Combinatorics \& Optimization, University of Waterloo, Waterloo, ON N2L 3G1,
Canada}

\abstract{The global Schwinger formula, introduced by Cachazo and Early as a single integral over the positive tropical Grassmannian, provides a way to uncover properties of scattering amplitudes which are hard to see in their standard Feynman diagram formulation. In a recent work, Cachazo and one of the authors extended the global Schwinger formula to general $\phi^p$ theories. When $p=4$, it was conjectured that the integral decomposes as a sum over cones which are in bijection with non-crossing chord diagrams, and further that these can be obtained by finding the zeroes of a piece-wise linear function, $H(x)$. In this note we give a proof of this conjecture. We also present a purely combinatorial way of computing $\phi^p$ amplitudes by triangulating a trivial extended version of non-crossing $(p-2)$-chord diagrams, called extended diagrams, and present a proof of the bijection between triangulated extended diagrams and Feynman diagrams when $p=4$. This is reminiscent of recent constructions using Stokes polytopes and accordiohedra. However, the $\phi^p$ amplitude is now partitioned by a new collection of objects, each of which characterizes a polyhedral cone in the positive tropical Grassmannian in the form of an associahedron or of an intersection of two associahedra. Moreover, we comment on the bijection between extended diagrams and double-ordered biadjoint scalar amplitudes. We also conjecture the form of the general piece-wise linear function, $H^{\phi^p}(x)$, whose zeroes generate the regions in which the $\phi^p$ global Schwinger formula decomposes into.}

\begin{document}
\maketitle
\addtocontents{toc}{\protect\setcounter{tocdepth}{1}}
\def \tr {\nonumber\\}
\def \la  {\langle}
\def \ra {\rangle}
\def\hset{\texttt{h}}
\def\gset{\texttt{g}}
\def\sset{\texttt{s}}
\def \be {\begin{equation}}
\def \ee {\end{equation}}
\def \ba {\begin{eqnarray}}
\def \ea {\end{eqnarray}}
\def \k {\kappa}
\def \h {\hbar}
\def \r {\rho}
\def \l {\lambda}
\def \be {\begin{equation}}
\def \en {\end{equation}}
\def \bes {\begin{eqnarray}}
\def \ens {\end{eqnarray}}
\def \red {\color{Maroon}}
\def \pt {{\rm PT}}
\def \s {\sigma} 
\def \ls {{\rm LS}}
\def \ma {\Upsilon}
\def \s {\textsf{s}}
\def \t {\textsf{t}}
\def \R {\textsf{R}}
\def \W {\textsf{W}}
\def \U {\textsf{U}}
\def \e {\textsf{e}}

\numberwithin{equation}{section}

\section{Introduction} \label{sec1}

Tropical geometry and tropical Grassmannians play an important role in various aspects of quantum field theory, as shown by a large amount of recent work on these topics, see \cite{Tourkine:2013rda,Cachazo:2019ngv,Cachazo:2022voc, Drummond:2019qjk, Cachazo:2021llu,Drummond:2019cxm,Henke:2019hve,Arkani-Hamed:2024vna,Arkani-Hamed:2022cqe,Lukowski:2020dpn,Early:2023tjj} for a few examples. The tropical Grassmannian ${\rm Trop}\, G(k,n)$ was first introduced by Speyer and Sturmfels in \cite{SSTrop}, where it was shown that it agreed with the moduli space of phylogenetic trees studied by Billera, Holmes and Vogtmann \cite{BHV}. Speyer and Williams then introduced positive tropical Grassmannians \cite{SWTrop}, ${\rm Trop}^+\, G(k,n)$, where ${\rm Trop}^+\, G(2,n)$ is related to the associahedron and parameterizes the space of planar trees on $n$ leaves. 

For example, figure \ref{gsf5} shows a positive part of the space of metric trees with $n=5$ leaves, which coincides with $\textrm{Trop}^+G(2,5)$ for the canonical ordering $\mathbb{I}:=(12345)$. The figure requires a bit of explanation. A point in each quadrant of this space is in correspondence with a metric tree with only degree $p=3$ vertices (shown in blue), where the internal lengths $\ell_i$ of the edges are interpreted as Schwinger parameters which are defined by the semi-rays in the quadrant. A point in each semi-ray is, therefore, in correspondence with a tree with a degree $p=3$ vertex, a degree $p=4$ vertex and a single edge. In this example, the black tree in the figure is in correspondence with the planar kinematic invariant $s_{23}:=(p_2+p_3)^2=(p_4+p_5+p_1)^2$. Notice that this tree in the figure can be obtained by collapsing the other edge (of length $\ell_2$ or $\ell_4$) of any of the two trees in the quadrants sharing the semi-ray $\ell_3$. This is due to the fact that any two quadrants sharing a semi-ray correspond to two trees related by a \textit{planar mutation} since, after contracting an edge, there is another possible way of opening up a new edge to generate a tree planar with respect to the given ordering. The pentagon in blue is the boundary of the dual associahedron with $n=5$ letters, also known as the \textit{link of the origin}, which is given by the union of all hypersurfaces $\sum_i\ell_i=1$ defined for each quadrant.

For general $n$ and for a given ordering $\alpha$, $\textrm{Trop}^+G(2,n)$ has C$_{n-2}$ positive orthants $(\mathbb{R}^+)^{n-3}$, where C$_m$ is the m$^{th}$ Catalan number, and this count coincides with the number of metric trees with $n$ leaves and only degree $p=3$ vertices which are planar with respect to such ordering $\alpha$. Codimension-$d$ boundaries of the space correspond to trees with $d$ zero-length edges and at least one degree $p>3$ vertex. 

\begin{figure}[h!]
	\centering
	\includegraphics[width=0.9\linewidth]{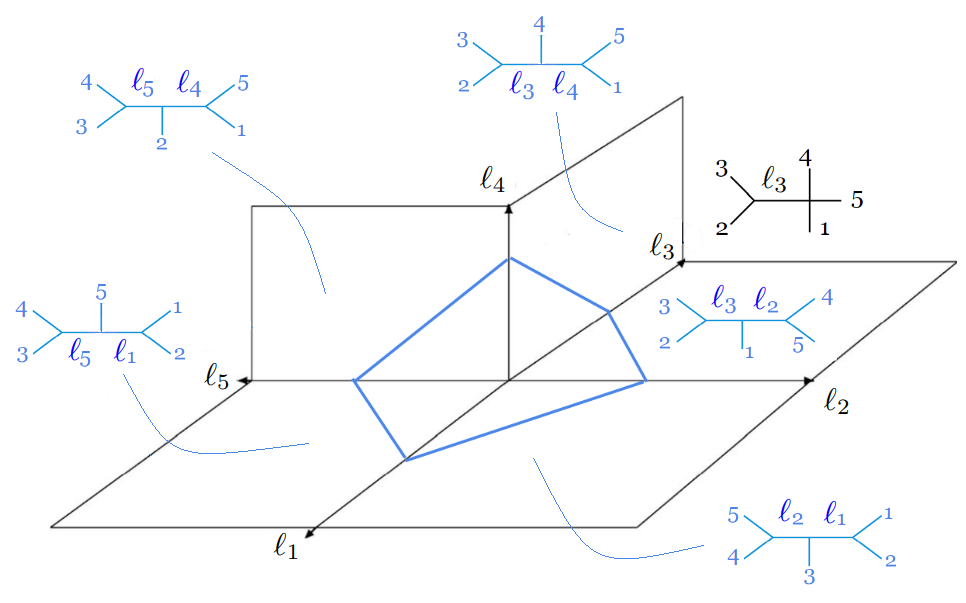}
	\caption{The space of planar metric trees with $n=5$ leaves.}
	\label{gsf5}
\end{figure}

A representative case in which $\textrm{Trop}^+G(2,n)$ plays an important role in physics is the realm of tree-level partial amplitudes of massless scalars in the biadjoint representation of a $U(N)\times U(\Tilde{N})$ flavor group with cubic interactions, $m_n(\alpha,\beta)$ \cite{Cachazo:2013iea}. These amplitudes can be computed by summing over all $n$-particle Feynman diagrams planar with respect to both orderings $\alpha$ and $\beta$. Therefore, when the two orderings coincide, e.g. $\alpha=\beta=\mathbb{I}$, the partial amplitude can be computed by summing over all degree $p=3$ metric trees with $n$ leaves which are planar with respect to $\mathbb{I}:=(12\cdots n)$. For example, when $n=5$ the partial amplitude $m_5(\mathbb{I},\mathbb{I})$ is given by the sum of the five blue trees in figure \ref{gsf5}. 

For the general $n$ case, the partial amplitude $m_n(\mathbb{I},\mathbb{I})$ is given by summing over all trees which are in correspondence with a positive orthant $(\mathbb{R}^+)^{n-3}$ in $\textrm{Trop}^+G(2,n)$. The contribution of a single tree $T$ can be computed as a Schwinger integral of the form
\be
{\cal R}(T)=\prod_{a=1}^{n-3}\int_0^{\infty}d\ell_{a}\textrm{exp}(-\ell_{a}t_{a})=\prod_{a=1}^{n-3}\frac{1}{t_{a}}\,,
\ee
where $\ell_{a}$ is the Schwinger parameter associated to the internal length of an edge, and $t_a$ is the kinematic invariant associated to that edge\footnote{The integral is only defined for $t_a>0$, but after integrating the resulting rational function is valid for any $t_a\neq0$.}. For example, the tree in figure \ref{gsf5} associated to the quadrant parametrized by $\ell_1$ and $\ell_2$ contributes as
$${\cal R}(T)=\int_0^{\infty}d\ell_1d\ell_2\textrm{exp}(-\ell_1 s_{12}-\ell_2 s_{45})=\frac{1}{s_{12}s_{45}}\,.$$

Having understood the connection between $\textrm{Trop}^+G(2,n)$ and partial biadjoint amplitudes, it is very tempting to try and attempt to unify the sum of the C$_{n-2}$ $(n-3)$-dimensional Schwinger integrals into a single object. Indeed, in a recent work Cachazo and Early showed that the partial amplitude $m_n(\mathbb{I},\mathbb{I})$ can be computed as a single integral over ${\rm Trop}^+\, G(2,n)$, known as the global Schwinger formula \cite{Cachazo:2020wgu} (see \cite{Cachazo:2022voc,Arkani-Hamed:2022cqe,GimenezUmbert:2023ykk,Arkani-Hamed:2023mvg,Arkani-Hamed:2024vna,Arkani-Hamed:2023lbd} for more work on global Schwinger parameterizations)

\be\label{amp2}
m_n(\mathbb{I},\mathbb{I}) = \int_{\mathbb{R}^{n-3}}\!\! d^{n-3}x\, {\rm exp}(-F_n(x))\,,
\ee 
which is determined by a piece-wise linear function $F_n(x)$ that we call the \textit{tropical potential}, thus providing a unification of all the Schwinger integrals. See figure \ref{gspn5} for a pictorial representation of the formula. 

Let us explain how to construct the tropical potential for future convenience. We start with a positive parameterization of $G^+(2,n)$ \cite{Alex} and then tropicalize the Pl\"ucker coordinates. A positive parameterization of $G^+(2,n)$ is given by\footnote{Where we have suppressed the torus coordinates.} 
\begin{align}\label{preT}
\begin{pmatrix}
1 & 0 & -1 & -(1+\tilde{x}_1) & -(1+\tilde{x}_1+\tilde{x}_2) & \cdots & -(1+\tilde{x}_1+\cdots+\tilde{x}_{n-3})\\
0 & 1 & 1 & 1 & 1 & \cdots & 1
\end{pmatrix}
\end{align}
where $\tilde{x}_a\in\mathbb{R}^+$ and any maximal minor $\Delta_{ab}$ with $a<b$ is positive. We then tropicalize a minor $\Delta_{ab}$ in \eqref{preT} by replacing addition, ${\tilde x_i}+{\tilde x_j}$, with the min-function ${\rm min}(x_i,x_j)$, and multiplication, ${\tilde x_i}{\tilde x_j}$, with addition $x_i+x_j$, with now the tropical variables being
unconstrained, i.e. $x_a\in\mathbb{R}$. The tropicalized minors evaluate to

\be\label{tropMin} 
\Delta^{\rm Trop}_{ab}(x) = \left\{ \begin{array}{ccc}
 {\rm min}(x_{a-2},x_{a-1},\ldots, x_{b-3}) &\,\,\,\,  & 2\leq a\leq b-1,\, 4\leq b \leq n \\
 0 & & {\rm otherwise}
\end{array}\right. , 
\ee 
where $x_0:=0$ and whenever there is a single argument we use ${\rm min}(x) := x$. The tropical potential function is therefore defined as\footnote{Here we are making the connection to the space of planar metric trees by identifying the matrix of distances $d_{ab}$ from leave $a$ to leave $b$ with the tropical Pl\"ucker coordinates $\Delta_{ab}^{\textrm{Trop}}$. See \cite{Cachazo:2022voc} for more details.}
\begin{align}\label{tropP}
    F_n(x):=\sum_{1\leq a < b \leq n}\!\!\! s_{ab}\, \Delta^{\rm Trop}_{ab}(x) =\sum_{b=4}^n\sum_{a=2}^{b-1}\! s_{ab}\,{\rm min}(x_{a-2},x_{a-1},\ldots, x_{b-3})\,,
\end{align}
where $s_{ab}:=(p_a+p_b)^2$ are kinematic invariants and $p_a$ is the momentum vector of particle $a$. One can also write the piece-wise linear function \eqref{tropP} in terms of planar kinematic invariants, $X_{a,b}:=(p_a+p_{a+1}+\cdots+p_{b-2}+p_{b-1})^2$, as

\be\label{PlanarF} 
F_n(x) = \sum_{a<b}X_{a,b+1}f_{[a,b]}(x)\,,
\ee 
where $s_{ab} = X_{a,b+1}-X_{a+1,b+1}-X_{a,b}+X_{a+1,b}$ and $X_{c,d+1}=0$ whenever $c\geq d$. The \textit{tropical cross-ratios} $f_{[a,b]}(x)$ are defined as

\be
f_{[a,b]}(x):= \Delta^{\rm Trop}_{a,b}(x)-\Delta^{\rm Trop}_{a,b+1}(x)-\Delta^{\rm Trop}_{a-1,b}(x)+\Delta^{\rm Trop}_{a-1,b+1}(x)\,.
\ee

\begin{figure}[h!]
	\centering
	\includegraphics[width=0.9\linewidth]{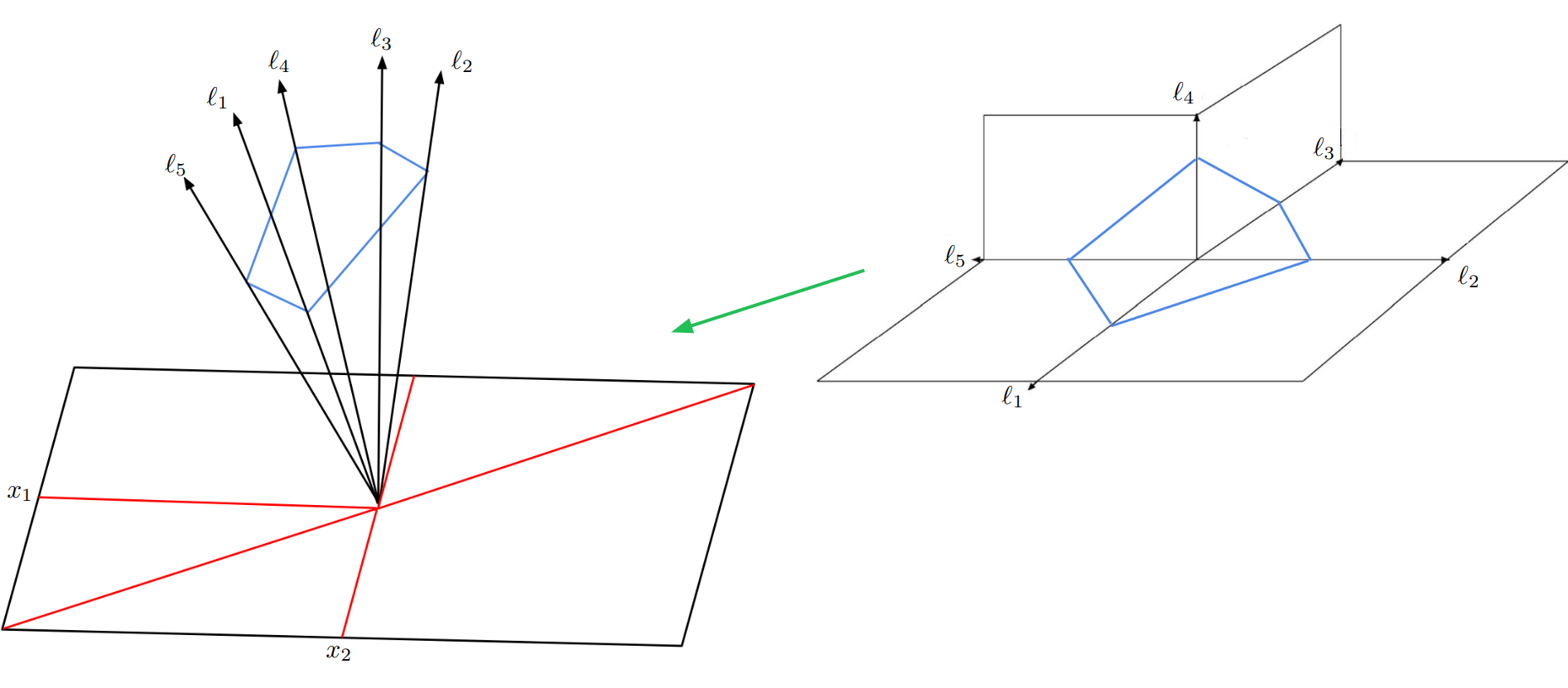}
	\caption{The global Schwinger parameterization for $n=5$, viewed as a projection of $\textrm{Trop}^+G(2,5)$, which unifies five 2-dimensional Schwinger integrals into a single 2-dimensional integral, parametrized by two tropical variables $x_1$ and $x_2$. The red lines on the plane define the domains where the tropical potential $F_5(x)$ becomes linear.}
	\label{gspn5}
\end{figure}
We refer the reader to \cite{Cachazo:2020wgu,Cachazo:2022voc} for more details and for comments on the existence and regions of convergence of the integral \eqref{amp2}.

At first it might seem puzzling that an integral over all ${\rm Trop}^+\, G(2,n)$ computes the partial amplitude $m_n(\mathbb{I},\mathbb{I})$, which only encodes trees with cubic interaction vertices, since ${\rm Trop}^+\, G(2,n)$ also contains trees with higher-degree vertices living in higher codimension boundaries of the space. However, the regions in ${\rm Trop}^+\, G(2,n)$ that contain higher-degree vertices are of measure zero and, therefore, do not contribute to the integral \eqref{amp2}. 

Recently, Cachazo and one of the authors continued the study of global Schwinger formulas for amplitudes in scalar theories with general tree-level $\phi^p$-interactions, for any $p>3$ \cite{Cachazo:2022voc}. Namely, these are amplitudes obtained by summing over trees with only degree-$p$ interaction vertices and which are planar with respect to one ordering (in this work we implicitly consider the canonical ordering $\mathbb{I}$ without loss of generality).

For the simplest example, i.e. $p=4$, the global Schwinger formula is given by

\be\label{amp56}  
A^{\phi^4}_n =  \int_{\mathbb{R}^{n-3}}\!\! d^{n-3}x\, {\rm exp}\left(-\sum_{b-a\equiv 0\,\rm{mod}\,2}X_{a,b+1}f_{[a,b]}(x)\right)Q(x)\,,
\ee 
where $Q(x)$ is a sum over distributions that localize the integral to the regions where trees with only degree-4 vertices live. In \cite{Cachazo:2022voc} it was conjectured that such regions, each of which corresponds to a polyhedral cone in ${\rm Trop}^+\, G(2,n)$ of dimension $n/2-1$ in $\mathbb{R}^{n-3}$, are given by the solutions to the equation $H(x)=0$, where $H(x)$ is

\begin{align}\label{genH}
H(x) = \sum_{a=0}^{n-3}x_a+2\sum_{a=0}^{n-4}\sum_{b=a+1}^{n-3}(-1)^{b-a}\,\textrm{min}(x_a,x_{a+1},...,x_b)\,.
\end{align} 
We give a proof of this conjecture in section \ref{sectriang4}. This result implies that $A_n^{\phi^4}$ amplitudes can be expressed as a sum over regions. Moreover, a remarkable connection to cubic amplitudes was found. Namely, each of the C$_{n/2-1}$ regions that define the support of the distribution $Q(x)$ is related to a double-ordered cubic amplitude $m_{n/2-1}(\alpha,\beta)$ with a smaller number of particles. Therefore, the $\phi^4$ amplitude can be expressed as a sum of products of cubic amplitudes \cite{Cachazo:2013iea}, something which is hard to see from the standard Feynman diagram formulation of the amplitude. In \cite{Cachazo:2022voc} it was also shown that each of these regions in $\textrm{Trop}^+G(2,n)$ are classified by non-crossing chord diagrams, which we review in section \ref{sectriang4}.

One of the main results of this paper, presented in section \ref{sectriang4}, is that we show that one can easily obtain the $\phi^4$ amplitude without ever having to carry out an integral by taking an extended version of the non-crossing chord diagrams, called extended diagrams, and then further decomposing them in a way which we call triangulation. In section \ref{sectriang4} we also present a formal proof of the bijection between triangulations of extended diagrams and $\phi^4$ trees, here using the fact that the extended diagram triangulations give quadrangulations of an even n-gon. Therefore, we proof that all triangulated extended diagrams account for the whole amplitude, showing that the whole collection of extended diagrams provides a novel way of partitioning $A_n^{\phi^4}$ into new objects, that differ e.g. from Stokes polytopes \cite{Banerjee:2018tun,Salvatori:2019phs,Aneesh:2019cvt,Srivastava:2020dly}, whose introduction was motivated by the connection between $\phi^3$ amplitudes and the associahedron \cite{Mizera:2017cqs,Arkani-Hamed:2017mur}. The reason this is different is that these objects are now given by regions in $\textrm{Trop}^+G(2,n)$ in the form of an associahedron or of an intersection of two associahedra. 

For example, figure \ref{p4n16intro} shows a triangulation of an extended diagram for $n=16$, together with the tree the triangulation is related to. Here every black chord separating two colored regions --which are related to cubic amplitudes-- and every triangulating chord --shown in red-- corresponds to a propagator in the tree. For example, the chord going from point 4 to point 9 on the real line gives rise to $1/X_{4,9}$, while the triangulating chord going from point 9 to point 14 gives rise to $1/X_{9,14}$.  In this correspondence a colored region with $m-1$ boundary chords will be related to an $m$-point amplitude and the further decomposition into smaller regions by the triangulating chords will be related to a triangulation of an $m$-gon for the $m$-point amplitude.  This why we use the term triangulation, even though these smaller regions each have $4$ boundary lines and arcs in the extended diagram.

In section \ref{sectriangp} we present an analogous procedure for general $\phi^p$ theories, using a more general version of extended diagrams, first introduced in \cite{Cachazo:2022voc}, which are also obtained using similar rules. Namely, we explain how we can obtain any $\phi^p$ amplitude by triangulating all extended diagrams.

\begin{figure}[H]
\includegraphics[width=15cm]{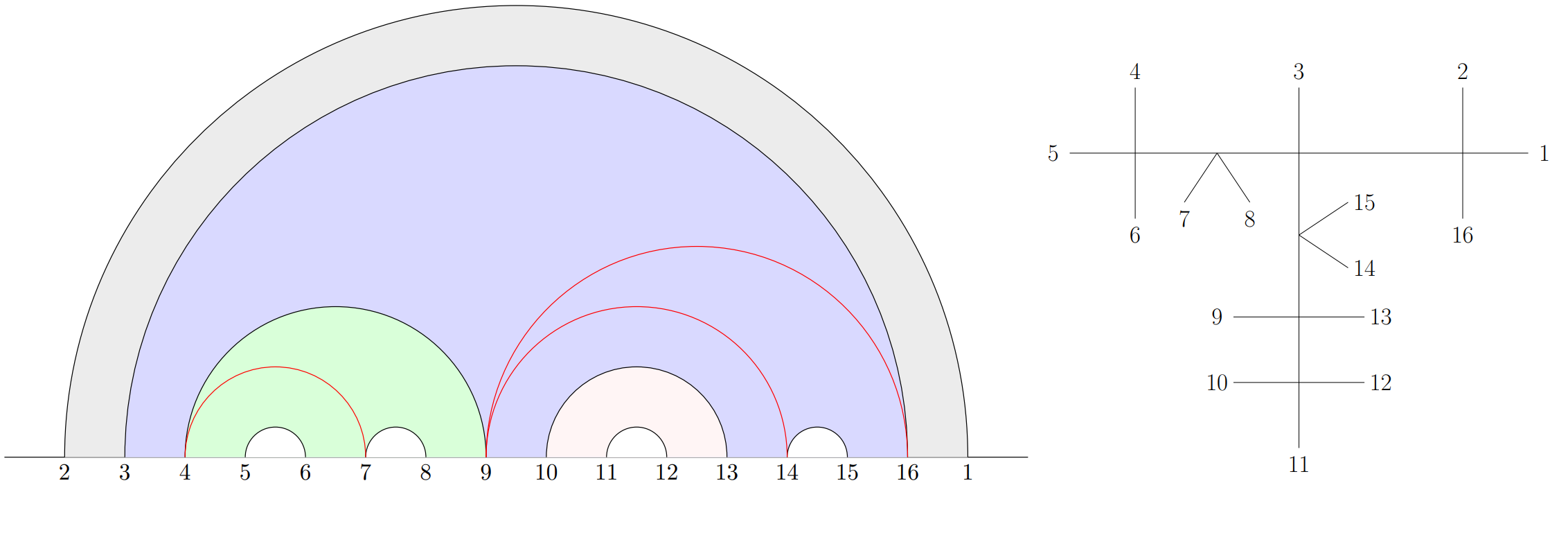}
\centering
\caption{Triangulation of an extended diagram for $n=16$ in $\phi^4$ and its associated tree. The triangulating chords are shown in red, and the colored regions correspond to different cubic subamplitudes.}
\label{p4n16intro}
\end{figure}
In section \ref{secp4cubic} we continue the exploration of the relation between extended diagrams, like the one in figure \ref{p4n16intro}, with double-ordered cubic amplitudes. We also comment on the fact that colored regions in an extended diagram must be related to a cubic subamplitude, by presenting a polygon decomposition of the extended diagram. We conclude in section \ref{discs} with some additional comments and discussions about future research topics. For example, we conjecture the analogous form of the function \eqref{genH} for general $p$, and we present some examples of polygon-decomposed extended diagrams in $\phi^p$ that require further exploration.

\section{$\phi^4$ from Triangulations of Extended Diagrams}\label{sectriang4}
A remarkable implication of the global Scwhinger formulation for $\phi^4$ theories found in \cite{Cachazo:2022voc} is that the amplitude decomposes as a sum over polyhedral cones in $\textrm{Trop}^+G(2,n)$, which are in bijection with non-crossing chord diagrams. Such diagrams provide a systematic way of finding the regions which turn into distributions in the integral \eqref{amp56}, ultimately computing the amplitude.

In this section we show one of the main results in this paper. Namely, that by triangulating a trivial extended version of these diagrams, known as extended non-crossing chord diagrams and originally introduced in \cite{Cachazo:2022voc}, we can obtain the $\phi^4$ amplitude without having to compute the integral.

Before showing how to obtain $\phi^4$ amplitudes from triangulations of extended diagrams, let us first review what these diagrams are.

\subsection{Non-Crossing Chord Diagrams and Extended Diagrams}
First of all, let us recall the definition of a non-crossing chord diagram, where we have slightly changed the labelling with respect to \cite{Cachazo:2022voc} for convenience.

\begin{defn}\label{chords}

Place $n-2$ points labeled $3,4, \ldots ,n$ in increasing order on the real line. A {\it non-crossing chord} diagram is a perfect matching of the points such that all edges, which we call \emph{chords}, can be drawn as arcs on the upper half plane without any crossings. Let us denote the chord connecting points $a$ and $b$ as $\theta_{ab}$.  
\end{defn}

This leads to the following conjecture which we will prove in section~\ref{proof}.
\begin{conj}[Conjecture 6.2 of \cite{Cachazo:2022voc}]\label{conj1}
The regions contributing to $A^{\phi^4}_n$ are in bijection with the set of all $\textrm{C}_{n/2-1}$ possible non-crossing chord diagrams. Moreover, the region $R$ corresponding to a particular diagram is obtained as follows:
\begin{itemize}
    \item For each chord $\theta_{ab}$ set $x_{a-3}=x_{b-3}$.
    \item If a chord $\theta_{ab}$ surrounds another chord $\theta_{cd}$, then $x_{a-3}=x_{b-3} < x_{c-3}=x_{d-3}$.
\end{itemize}
In other words, the regions defined by non-crossing chord diagrams are all the solutions to $H(x)=0$, where $H(x)$ is given by equation \eqref{genH}.
\end{conj}
For example, in $n=6$ there are two non-crossing chord diagrams as shown in figure \ref{fig:n6_regions4}.

\begin{figure}[h!]
	\centering
	\includegraphics[width=0.90\linewidth]{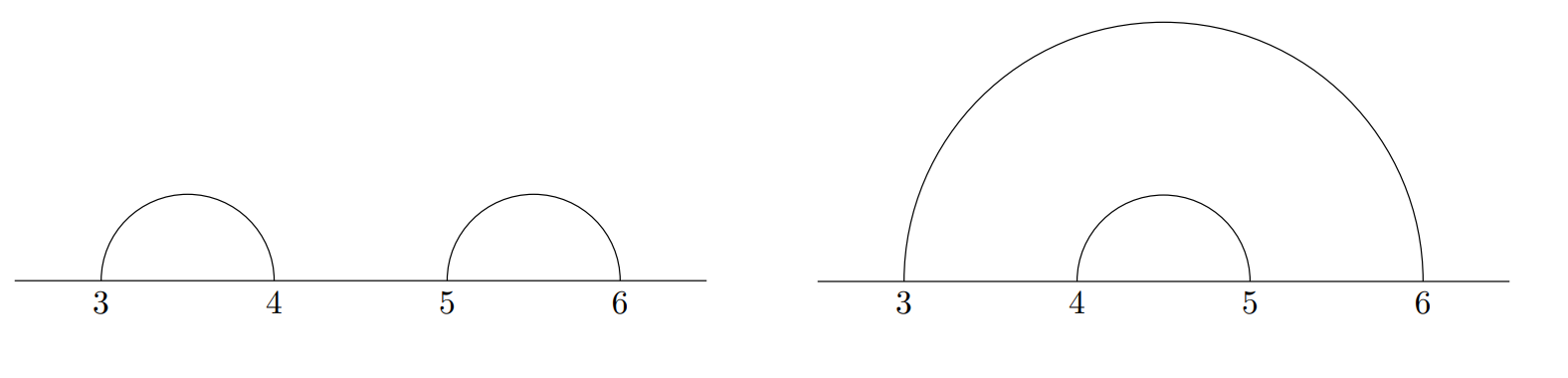}
	\caption{Non-crossing chord diagrams for $n=6$. On the right, the chord $\theta_{36}$ surrounds the chord $\theta_{45}$ and therefore the condition $x_0<x_1$ is imposed.}
	\label{fig:n6_regions4}
\end{figure}
The corresponding regions are
\be
R_1 = \{ x_0 = x_1,\, x_2 = x_3 \}\,, \quad R_2 = \{ x_0=x_3,\, x_1=x_2,\, x_0 < x_1 \}\,. 
\ee 
By evaluating \eqref{amp56} in each of these regions one finds

\be
A_6^{\phi^4:R_1}=\frac{1}{X_{1,4}}+\frac{1}{X_{2,5}}\,, \hspace{15mm} A_6^{\phi^4:R_2}=\frac{1}{X_{3,6}}\,,\label{amp46reg}
\ee
and one obtains the amplitude $A_6^{\phi^4}$ by summing over both regions.

Similarly, for $n=8$ there are five possible non-crossing chord diagrams, shown in figure \ref{figncn8}, which give rise to the five regions
\begin{align}\label{eight4}
    R_1 = \,& \{ x_1=0,\, x_2 = x_3,\, x_4= x_5 \}\,,\hspace{5mm}R_2 =\{  x_1=0, \, x_2=x_5,\, x_3=x_4,\, x_2 < x_3 \}\,, \nonumber \\ 
    R_3 = \,& \{ x_5= 0,\, x_1=x_2,\, x_3=x_4,\, x_1 > 0,\, x_3 > 0 \}\,, \nonumber \\ 
    R_4 = \,& \{ x_3 =0,\, x_1=x_2,\, x_4=x_5,\, x_1>0 \} \,,\hspace{1mm}
    R_5 =\{ x_5 = 0 ,\, x_1=x_4,\, x_2=x_3,\, x_2>x_1>0 \}\,,
\end{align} 
each of which contributes to the $\phi^4$ amplitude after using the integral \eqref{amp56}, as
\begin{align}\label{qmet}
A^{\phi^4: (1)}_8 = \,&  \frac{1}{X_{1,4}X_{4,7}}+\frac{1}{X_{4,7}X_{2,7}}+\frac{1}{X_{2,7}X_{2,5}}+\frac{1}{X_{2,5}X_{1,6}}+\frac{1}{X_{1,6}X_{1,4}}\,, \nonumber \\
A^{\phi^4: (2)}_8 = \,&  \frac{1}{X_{5,8}}\left(\frac{1}{X_{1,4}}+\frac{1}{X_{2,5}}\right)\,,\hspace{1mm}
A^{\phi^4: (3)}_8 = \frac{1}{X_{3,8}}\left(\frac{1}{X_{3,6}}+\frac{1}{X_{5,8}}\right)\,,\nonumber \\
A^{\phi^4: (4)}_8 = \,&  \frac{1}{X_{3,6}}\left(\frac{1}{X_{1,6}}+\frac{1}{X_{2,7}}\right)\,,\hspace{1mm}
A^{\phi^4: (5)}_8 =\frac{1}{X_{4,7}X_{3,8}}\,.
\end{align} 

\begin{figure}[h!]
	\centering
	\includegraphics[width=0.95\linewidth]{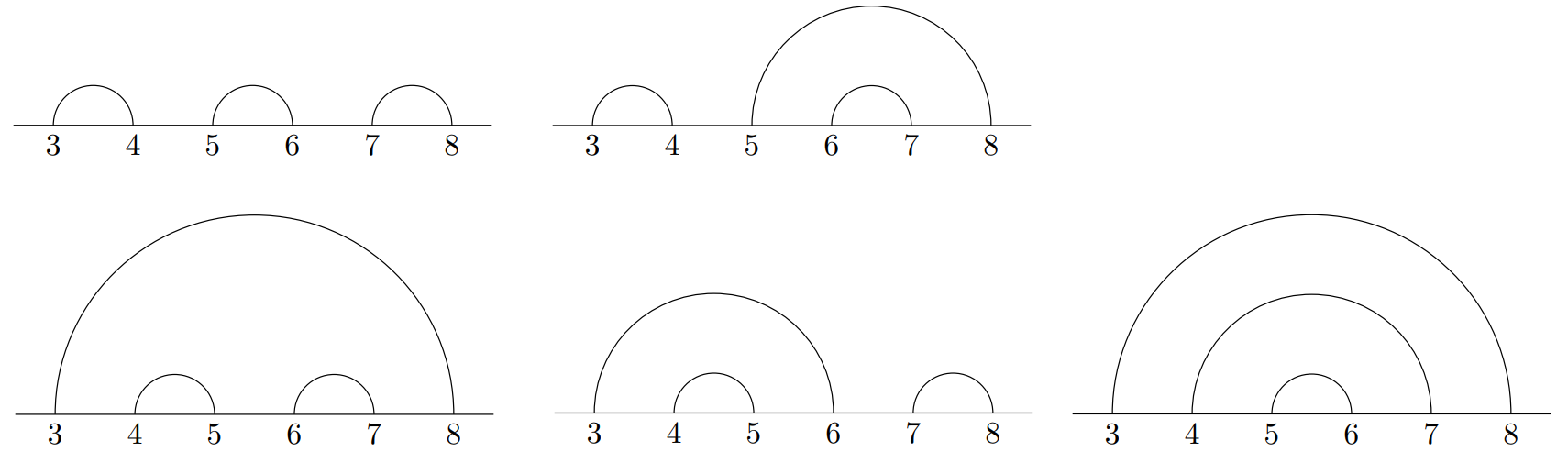}
	\caption{Non-crossing chord diagrams for $n=8$. In the second diagram $\theta_{58}$ surrounds $\theta_{67}$ and, therefore, $x_2<x_3$. In the third diagram $\theta_{38}$ surrounds both $\theta_{45}$ and $\theta_{67}$ so we have $x_0<x_1$ and $x_0<x_3$. In the fourth diagram $\theta_{36}$ surrounds $\theta_{45}$, so we have $x_0<x_1$. Finally, in the fifth diagram $\theta_{38}$ surrounds $\theta_{47}$, which at the same time surrounds $\theta_{56}$, so $x_0<x_1<x_2$.}
	\label{figncn8}
\end{figure}

The reader familiar with the biadjoint scalar theory will notice from these examples that the contribution to every region resembles that of $m_{n/2+1}(\alpha ,\beta)$ for some permutations $\alpha$ and $\beta$. This is an example of how the global Schwinger formula uncovers properties of scattering amplitudes which are non-obvious from the standard Feynman diagram formulation. In section \ref{secp4cubic} we explore this connection. 

Now we proceed to define the central object of this paper: the extended non-crossing chord diagram, which from now on we will call \textit{extended diagram} for simplicity.

\begin{defn}
An extended diagram associated to $A^{\phi^4}_n$ is a non-crossing chord diagram on $n$ points labeled by $\{ 2,3,4,5,\ldots ,n,1 \}$ in which the chord $\theta_{21}$ is always included. We also define a meadow of an extended diagram as any region in the diagram delimited by more than one chord and by the line on which the points lie.  A meadow is an $m$-point meadow if there are $m-1$ chords delimiting it. 
\end{defn}
In other words, an extended diagram is a non-crossing chord diagram with an additional chord, which goes from 2 to 1, surrounding all other chords. In the combinatorial language of chord diagrams, these are known as \textit{indecomposible} non-crossing chord diagrams.  Figures \ref{extp4n6m} and \ref{extp4n8m} show all the extended diagrams for $n=6$ and $n=8$, respectively. The reason why we have colored every meadow will become clear in section \ref{secp4cubic}, but for now it serves us as a visual way to differentiate between them.

\begin{figure}[h!]
	\centering
	\includegraphics[width=0.95\linewidth]{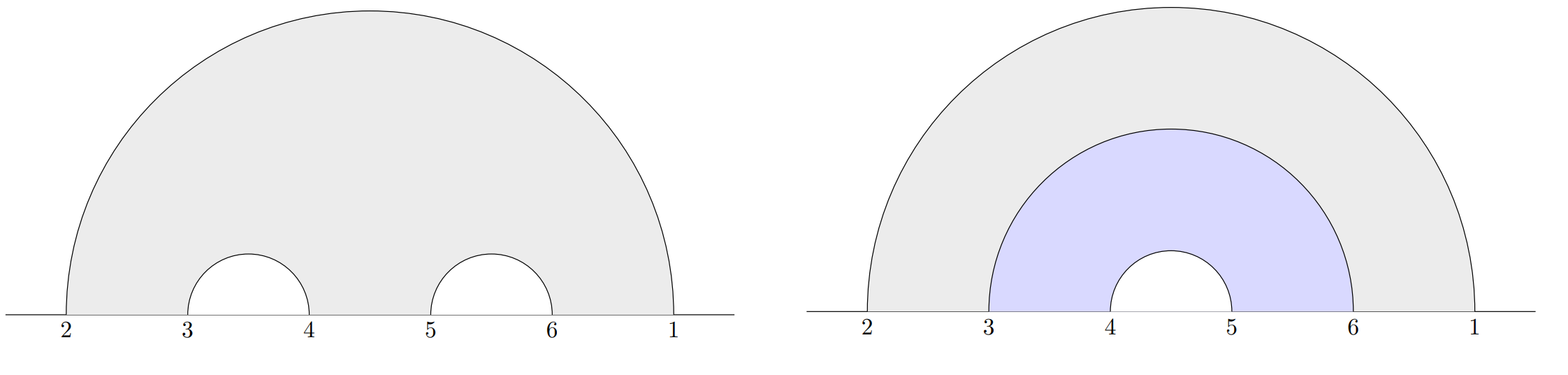}
	\caption{Extended diagrams for $n=6$. Every colored region in the diagram corresponds to a meadow.}
	\label{extp4n6m}
\end{figure}

\begin{figure}[h!]
	\centering
	\includegraphics[width=0.95\linewidth]{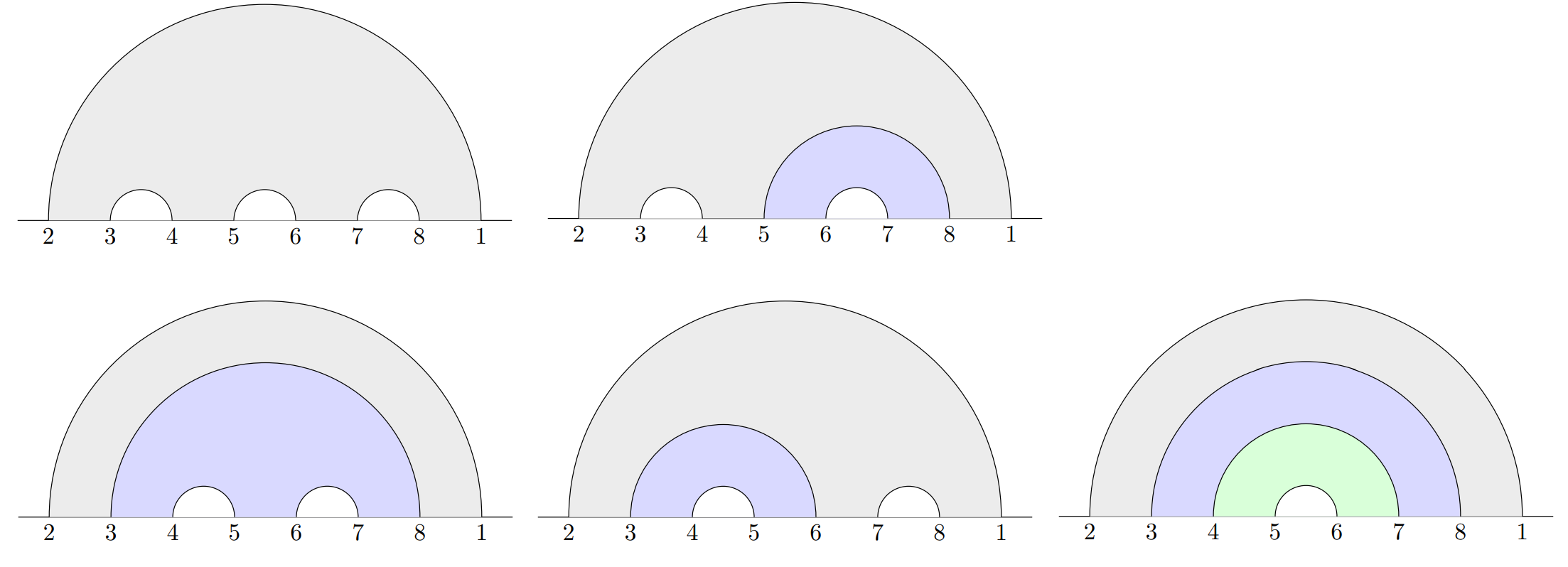}
	\caption{Extended diagrams for $n=8$. Every colored region in the diagram corresponds to a meadow.}
	\label{extp4n8m}
\end{figure}

\subsection{$\phi^4$ from Triangulations of Extended Diagrams}
In \cite{Cachazo:2022voc} it was shown that there are C$_{n/2-1}$ extended diagrams that contribute to the $\phi^4$ amplitude. In this section we show that triangulations of extended diagrams, which are counted by Fuss-Catalan numbers\footnote{Here $\textrm{FC}_m(q,r)$ is the Fuss-Catalan number given by $$\textrm{FC}_m(q,r)\equiv \frac{r}{mq+r} {mq+r \choose m}\,.$$} FC$_{n/2-1}(3,1)$, correspond to $\phi^4$ trees participating in the amplitude. This means that every extended diagram contains at least one $\phi^4$ tree. Therefore, one obtains the amplitude by summing over triangulations.

A triangulation of an extended diagram is given by adding new chords $\Omega_{ab}$ in the following way:
\begin{itemize}
    \item A triangulating chord $\Omega_{ab}$ must start, reading from left to right, from the label $a$ on the right of a lower chord of the meadow or on the left of the upper chord of the meadow, and must end on the label $b$ on the left of a lower chord of the meadow or on the right of the upper chord of the meadow.
    \item A triangulating chord must surround at least another chord in the diagram. 
    \item The triangulation ends when only regions with 4 delimiting chords and lines are left. 
\end{itemize}
Once the triangulation of an extended diagram is done, one can easily compute its contribution to the $\phi^4$ amplitude in the following way:
\begin{itemize}
    \item Every already existing chord $\theta_{ab}$ in the diagram surrounding another chord, with the exception of $\theta_{21}$, and every triangulating chord $\Omega_{ab}$, correspond to a propagator of the form $1/X_{a,b}$, where $X_{a,b}=X_{b,a}$.
    \item The contribution of a triangulated diagram to the amplitude is given by multiplying all propagators involved.
\end{itemize}
In the following subsection we show why every triangulated diagram is in bijection with a $\phi^4$ tree. It is also worth mentioning that the distinction between two kinds of chords, $\theta_{ab}$ and $\Omega_{ab}$, will become more clear in section \ref{secp4cubic}.

For example, let us consider the first diagram in figure \ref{extp4n6m}. The two possible triangulations of this diagram are shown in figure \ref{trin61}, where the triangulating chords are shown in red.

\begin{figure}[H]
\includegraphics[width=15cm]{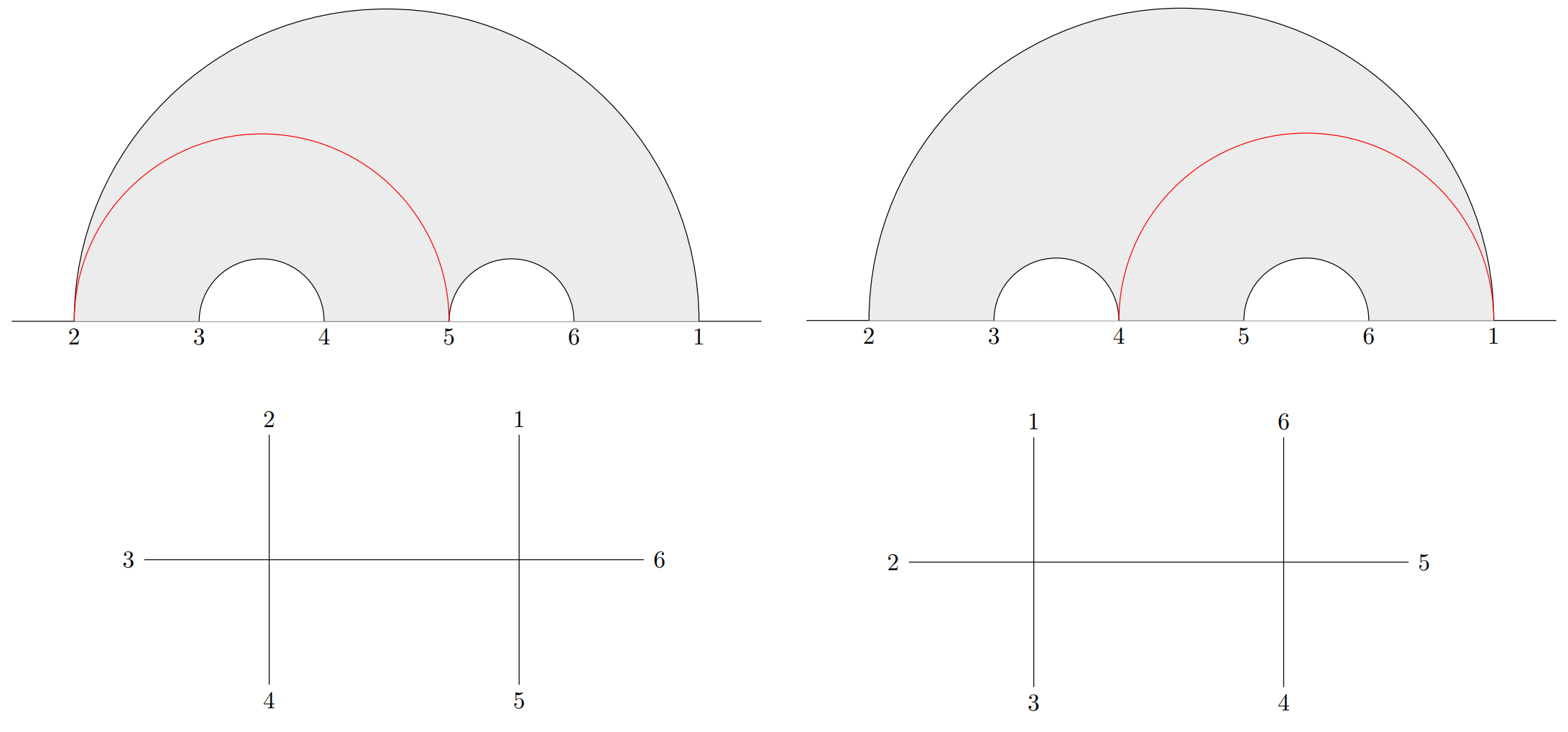}
\centering
\caption{Triangulations of an extended diagram for $n=6$ and the $\phi^4$ tree they correspond to, shown below.}
\label{trin61}
\end{figure}
The diagram on the left has one triangulating chord $\Omega_{25}$, thus it corresponds to a propagator of the form $1/X_{2,5}$, which is the contribution to the amplitude from the tree below. The diagram on the right has one triangulating chord $\Omega_{41}$, and it therefore contributes as $1/X_{1,4}$, the same as the tree below. The sum of the two triangulations gives $A_6^{\phi^4:R_1}$ in \eqref{amp46reg}.

The second diagram in figure \ref{extp4n6m} is already completely triangulated, i.e.\ it is only made of 3-point meadows. This diagram contains a chord of the form $\theta_{36}$ and it therefore contributes as $1/X_{3,6}$, which is the contribution of the remaining tree in the $\phi^4$ amplitude, i.e. $A_6^{\phi^4:R_2}$ in \eqref{amp46reg}. This is the only valid chord contributing to the amplitude in the diagram, since the remaining chords are $\theta_{21}$ and one that does not surround any other chord. Therefore, one gets the total amplitude by summing over triangulated extended diagrams, which gives the same result as summing over regions like in \eqref{amp46reg}.

Now let us move on to a more interesting example, when $n=8$. Consider the first extended diagram in figure \ref{extp4n8m}, i.e. the one with a single meadow. It has five possible triangulations, which are represented in figure \ref{trin81}, again with the triangulating chords in red, together with their associated $\phi^4$ trees below.

\begin{figure}[H]
\includegraphics[width=15cm]{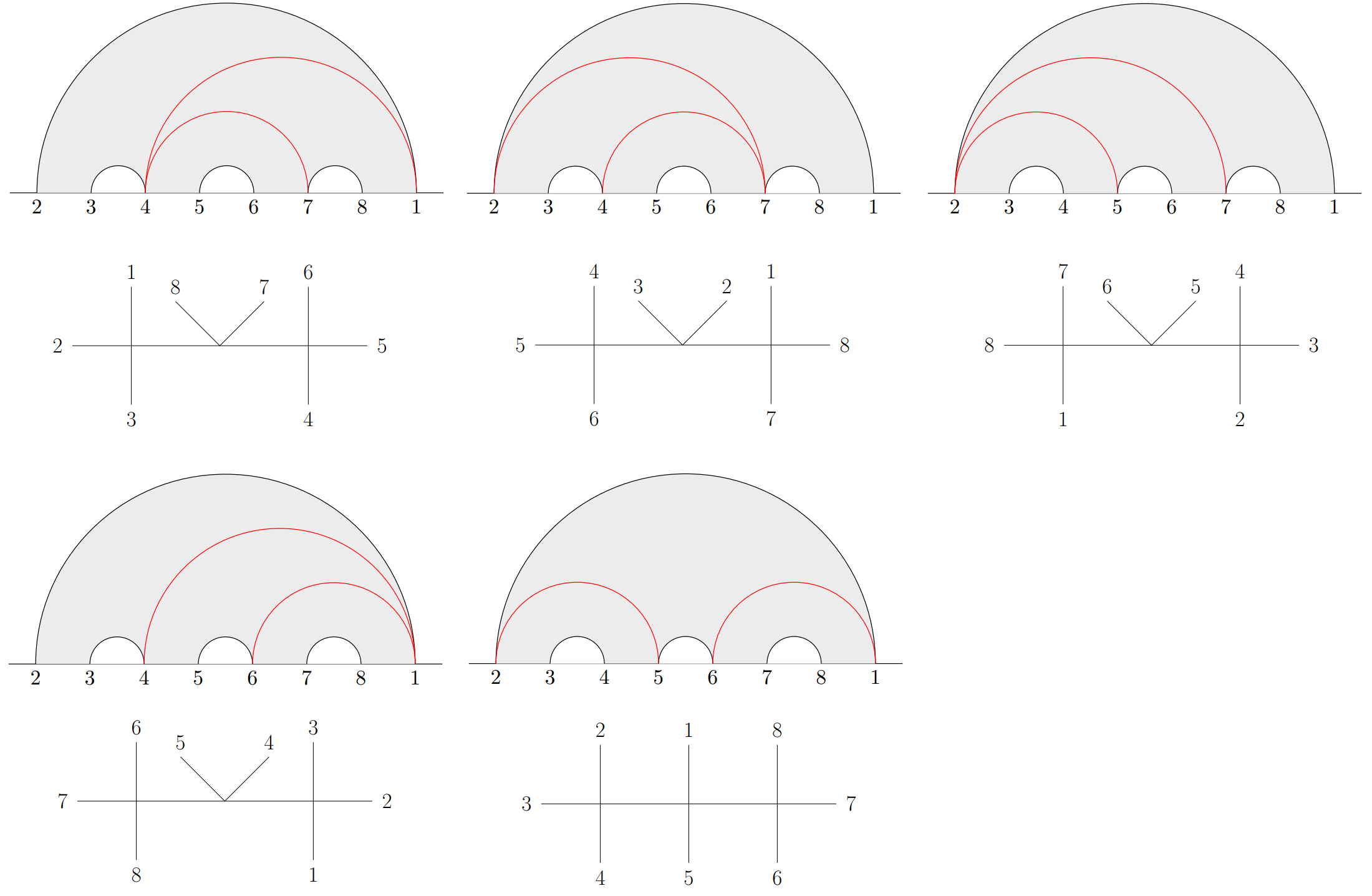}
\centering
\caption{Triangulations of an extended diagram for $n=8$ and their corresponding $\phi^4$ tree.}
\label{trin81}
\end{figure}
For example, the triangulated diagram on the bottom-right of the figure has two chords of the form $\Omega_{25}$ and $\Omega_{61}$ which correspond to propagators of the form $1/X_{2,5}$ and $/X_{1,6}$, respectively. Therefore, this triangulation contributes as
$$\frac{1}{X_{2,5}X_{1,6}}\,,$$
the same as the tree below. Summing over these five triangulations gives $A_8^{\phi^4:(1)}$ in \eqref{qmet}. The attentive reader might have noticed that this extended diagram only consists of a 5-point meadow, and the five triangulations are in bijection with those of a pentagon. This relation
will become clear in section \ref{secp4cubic}.

\begin{figure}[H]
\includegraphics[width=15cm]{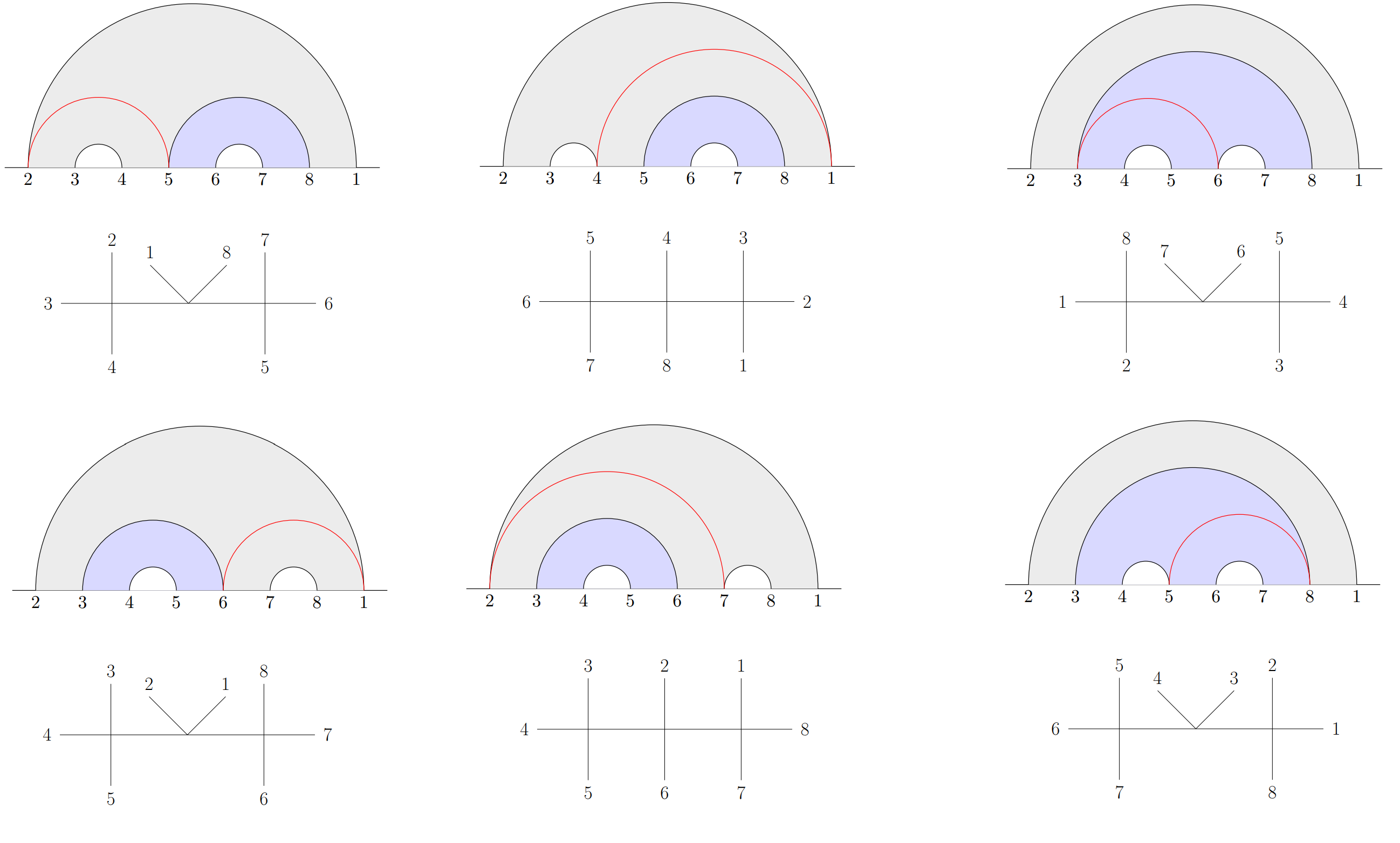}
\centering
\caption{Triangulations for the rest of the extended diagrams in $n=8$, and their corresponding $\phi^4$ tree shown below, with the exception of the extended diagram that only has 3-point meadows.}
\label{trin82}
\end{figure}

Figure \ref{trin82} shows the triangulations for other extended diagrams in $n=8$, with their corresponding trees. For example, the second extended diagram in figure \ref{extp4n8m} has a 4-point meadow and a 3-point meadow. Its triangulations are the first two in the first row of figure \ref{trin82}, which are those of a square. In this case, the first triangulated diagram has two contributing chords $\theta_{58}$ and $\Omega_{25}$, which give rise to a contribution of the form $1/X_{5,8}\times1/X_{2,5}$. The second triangulated diagram has two contributing chords $\Omega_{41}$ and $\theta_{58}$, which give rise to a contribution of the form $1/X_{1,4}\times1/X_{5,8}$. Therefore, the contribution from this extended diagram is given by 
$$\frac{1}{X_{5,8}}\left(\frac{1}{X_{1,4}}+\frac{1}{X_{2,5}}\right)\,,$$
which corresponds to $A_8^{\phi^4:(2)}$ in \eqref{qmet} and is in bijection with a double-ordered 5-point cubic amplitude. One can perform a similar analysis for the rest of the diagrams in figure \ref{trin82}.

Notice that the last diagram in figure \ref{extp4n8m} is already completely triangulated, this is why we have not included it in figure \ref{trin82}, but it of course corresponds to the remaining tree, giving rise to $1/X_{3,8}\times1/X_{4,7}$, which corresponds to $A_8^{\phi^4:(5)}$ in \eqref{qmet}. The sum of all the 12 triangulations, which are in bijection with the 12 $\phi^4$ trees, gives the total amplitude found by summing the terms in \eqref{qmet}.

As a last example, let us look at the triangulated extended diagram in figure \ref{p4n16}, the same presented in section \ref{sec1}, which corresponds to a $\phi^4$ tree in the $n=16$ case.

\begin{figure}[H]
\includegraphics[width=15cm]{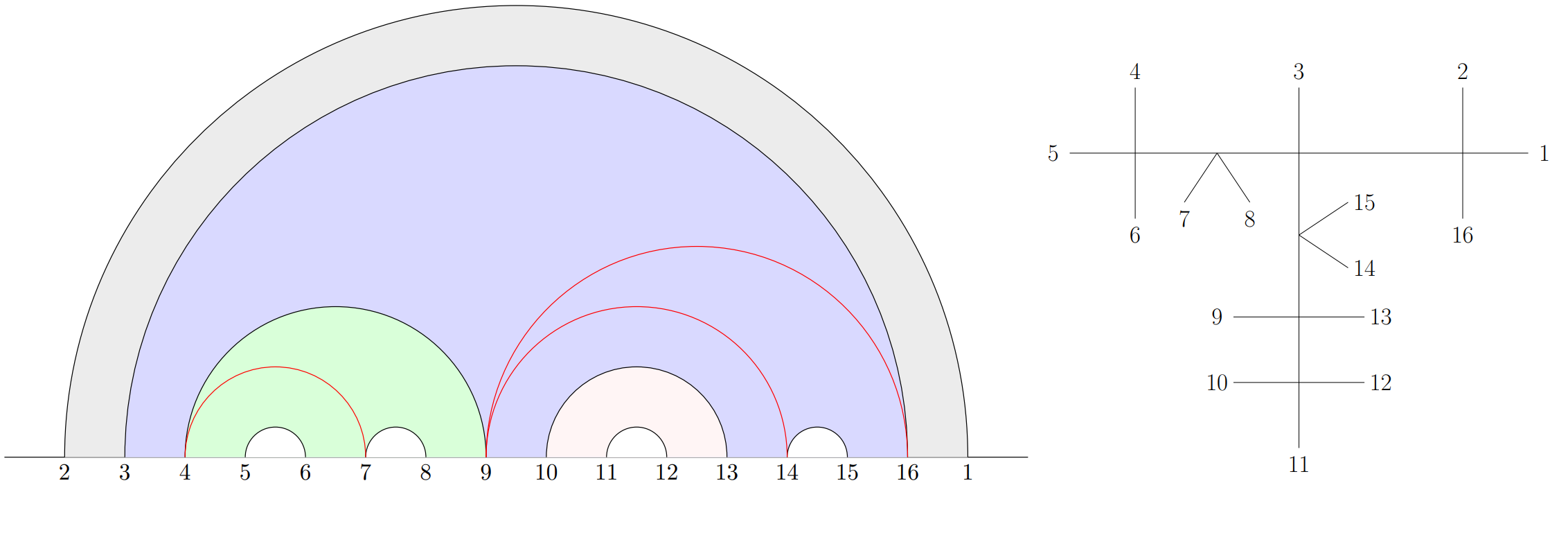}
\centering
\caption{Triangulation of an extended diagram for $n=16$ and their corresponding $\phi^4$ tree. The triangulating chords are shown in red, and the colored regions correspond to different meadows of the original extended diagram.}
\label{p4n16}
\end{figure}
This triangulation has chords of the form $\theta_{3,16}$, $\theta_{49}$, $\theta_{10,13}$, $\Omega_{47}$, $\Omega_{9,14}$ and $\Omega_{9,16}$ which correspond to propagators of the form $1/X_{3,16}$, $1/X_{4,9}$, $1/X_{10,13}$, $1/X_{4,7}$, $1/X_{9,14}$ and $1/X_{9,16}$, respectively. Therefore, this triangulation contributes as 
$$\frac{1}{X_{3,16}X_{4,9}X_{10,13}X_{4,7}X_{9,14}X_{9,16}}\,$$
the same as the tree shown in figure \ref{p4n16}. In fact, this triangulation comes from an extended diagram with two 3-point meadows --colored in grey and pink--, a 5-point meadow --colored in blue-- and a 4-point meadow colored in green. The triangulation is equivalent as that of a square times that of a pentagon, times three more propagators, giving the 6 propagators in the term. This comes from the fact that this extended diagram is in bijection with a 9-point double ordered amplitude, as wee will explain in section \ref{secp4cubic}.

\subsection{Bijection Betwen Triangulations of Extended Diagrams and $\phi^4$ Trees}\label{sec bijection}
Let us now write the main result of this construction more formally.
\begin{thm}
    For each $n\geq 2$ the set of triangulations of extended diagrams is in bijection with the set of 4-valent planar trees with $n$ leaves with leaves labelled cyclically $1,2,\ldots, n$ according to the planar structure.  Consequently, the extended diagrams give a partition of the $\phi^4$ amplitude by collecting tree-level Feynman diagrams according to the extended diagram whose triangulation gives the tree under this bijection.
\end{thm}

\begin{figure}[H]
\includegraphics[width=15cm]{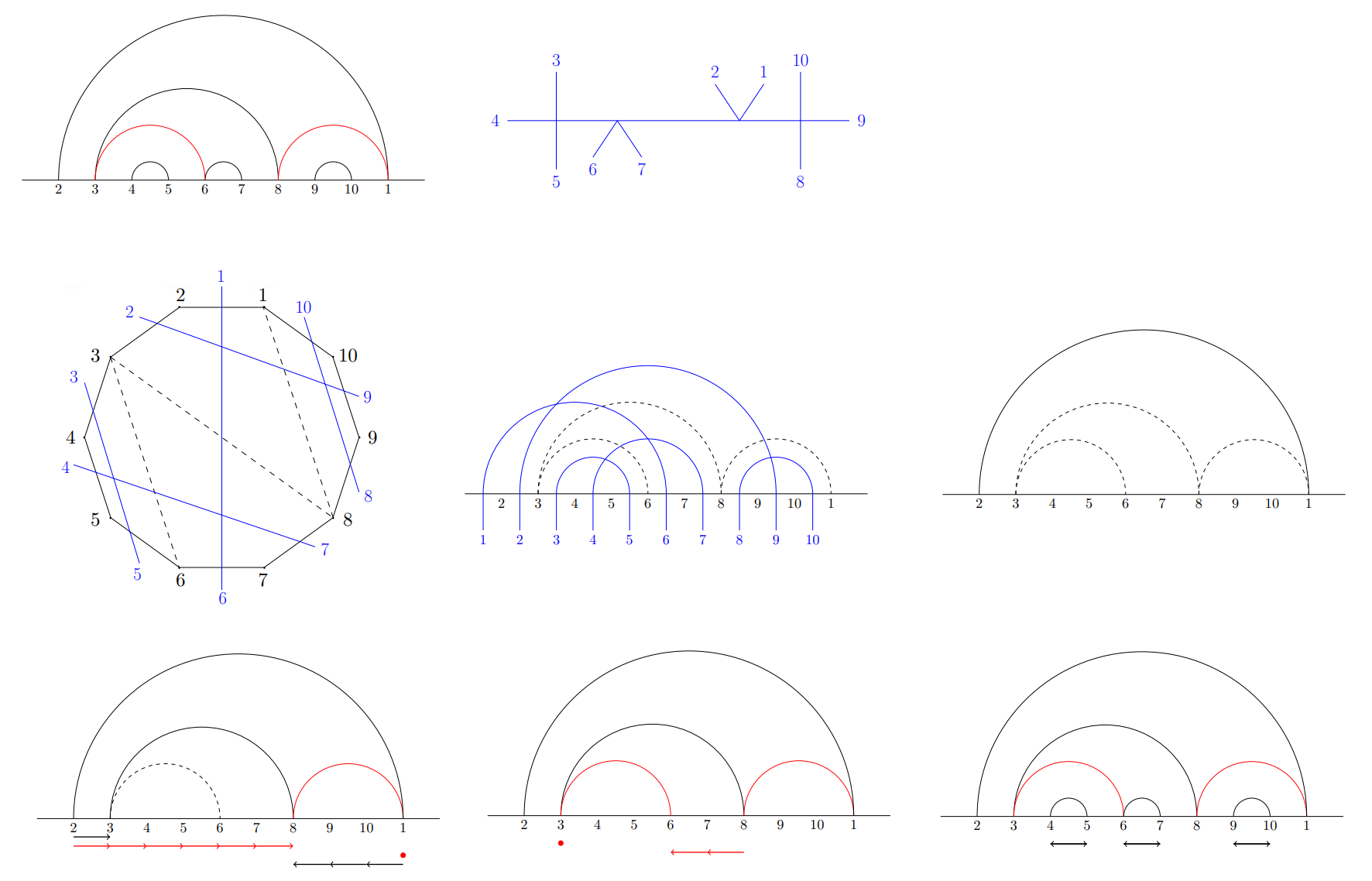}
\centering
\caption{Sequence of steps presented in the proof. \textit{First row}: triangulated extended diagram for $n=10$ together with the $\phi^4$ tree that it corresponds to. \textit{Second row}: the quadrangulation of the 10-gon that is associated to the tree, also represented after cutting edge 1. On the right we have included a black chord from 2 to 1. Uncolored chords are represented as dashed lines. \textit{Third row}: we color the chords using the procedure explained in the proof. Sequences of black arrows count an odd number of steps, while sequences of red arrows count an even number of steps, including zero steps (represented as a dot). Double sided arrows indicate pairs of consecutive vertices with no black chords as endpoints, thus leading to black chords joining them.}
\label{picproof}
\end{figure}

\begin{proof}
    We will prove the bijection by exhibiting the function, the purported inverse, and then checking the purported inverse is a two sided-inverse, which implies the injectivity and surjectivity of both functions.

    The function from triangulations of extended diagrams to leaf-labelled planar 4-valent trees is as described above.  It is clear that this function is well-defined and maps to the correct set with $n$ as described in the statement.

    The key step is to describe the inverse function.  Dual to a planar 4-valent tree with $n$ leaves cyclically labelled $1, 2, \ldots, n$ is a quadrangulation of an $n$-gon with the edges cyclically labelled $1,2,\ldots, n$.  Shifting each label clockwise to the neighbouring vertex and then cutting the edge between vertex $1$ and $2$ (which before the label shift had been edge $1$), we can draw the remaining edges of the $n$-gon horizontally on a line with vertices labelled $2, 3, \ldots, n, 1$ from left to right and the internal edge of the quadrangulation becoming arcs.  This diagram is the triangulated extended diagram with two exceptions.  First, the chords of length 2 are missing and second the chords which will be internal chords in the extended diagram must be distinguished from the triangulating chords.

    This will be done by a recursive parity rule.  For convenience we will call the chords of the extended diagram (whether internal or external) \emph{black} and the triangulating chords \emph{red}, lining up with the conventions in our figures.

    First add the chord from $2$ to $1$ as a black chord.  Call $2$ the left root and $1$ the right root.  Count inwards from the left root until reaching a vertex incident to a chord that is not yet colored.  If an even number of steps (including potentially no steps) were taken to reach this vertex then color all the uncolored chords at this vertex red and repeat.  If an odd number of steps were taken to reach this vertex then color the outermost chord at this vertex black, jump to the other end of this black chord, counting this as one step, and continue.  Stop when the right root is met. Likewise count inwards from the right root until reaching a vertex incident to a chord that is not yet colored, coloring and proceeding as in the other case but with the stopping condition being if the right root is met.

    We claim that this process is unambiguous in the sense that if a chord is reached from both directions then it will get the same color from each side.  The claim holds because in a quadrangulation the two vertices at the ends of any chord must have opposite parity, but our left and right root also had opposite parity, so counting from each one the parity to the nearest ends of another chord will be the same.  Furthermore, jumping from one end of a chord to another changes the parity, so counting it as one step when jumping across a black chord preserves parity considerations.

    For each chord colored black in the process above, repeat the process on the region inside this chord with the leftmost endpoint of the chord as the left root and the rightmost endpoint of this chord as the left root.

    Recursively this colors all the chords either red or black.  Furthermore, the black chords have a nested structure in which the black chords immediately inside a particular black chord have their left and right endpoints an odd number of steps from the left and right endpoints of the surrounding black chord.  This implies two things, first no black chords share an endpoint, and second consecutive sequences of vertices with no endpoint of a black chord must be of even length.  For each such sequence of vertices, add a black chord to the first and second vertex, the third and fourth vertex, and so on.

    This describes the purported inverse function.  Now we will make a few remarks about it.  By the observations of the previous paragraph, the black chords define a chord diagram, and by construction there is a chord from $2$ to $1$ so this is an extended diagram in the sense of this paper.  By the parity conditions on red edges, the red edges are valid triangulation chords and the result is a quadrangulation so all the meadows are 3-point meadows.  Therefore, the purported inverse function does define a triangulated extended diagram.

    It remains to show that the functions are inverses in both directions.  Beginning with a tree, taking the purported inverse as described above and then taking the original funtion back to trees clearly gives the identity because the quadrangulation dual to the tree remains intact throughout, with the chords of length 2 not affecting the tree in the original function.  The other direction is a bit more intricate.  Begin with a triangulated extended diagram.  Taking the original function to obtain a tree and then applying the purported inverse function, we have the same underlying quadrangulation since it is simply the dual to the tree.  What we need to check is that the coloring of the chords is consistent with the original and that the size 2 chords are added in a way that is consistent with the original.  Since the internal black chords determine the length two chords, it suffices to check that the internal chords are correctly colored.  This we prove following the recursive structure of the map (which could be set up as a formal induction should the reader care to). The other black chords of the meadow touching the chord from $2$ to $1$ must all have an odd left endpoint and an even right endpoint, and so are correctly colored black in the first iteration of the algorithm defining the purported inverse, while all the other chords in the meadow are correctly colored red.  Now looking at each meadow adjacent to this meadow notice that the parity of left and right endpoints has flipped, both of the outer chord and the inner ones.  This leaves the relative parity unchanged and so the algorithm correctly colors the chords in this meadow and so on.

    This proves that the purported inverse really is a two-sided inverse to the original map and hence completes the proof of the theorem.
\end{proof}

In this section we have seen how to compute $\phi^4$ amplitudes just by triangulating the extended diagrams that characterize the regions in $\textrm{Trop}^+G(2,n)$ that turn into distributions in the global Schwinger formula for $\phi^4$. In the next section we show how to analogously compute $\phi^p$ amplitudes for general $p$, by triangulating extended non-crossing $(p-2)$-chord diagrams. But before finishing with this section, we proceed to proof Conjecture \ref{conj1}.

\subsection{Proof of Conjecture~\ref{conj1}}\label{proof}

The aim of this section is to prove Conjecture~\ref{conj1}.
We need some preliminary definitions and lemmas.

\begin{defn}\label{def ki li}
  Let $x_0, x_1, \ldots, x_{n-3}$ be totally ordered.  Write 
  \begin{align*}
       k_i & = \max\{j: x_{i-\ell} \geq x_i \text{ for } 0 \leq \ell \leq i \}, \\
       \ell_i & = \max\{j: x_{i+\ell} \geq x_i \text{ for } 0 \leq \ell \leq i\}.
\end{align*}
\end{defn}

The idea of $k_i$ and $\ell_i$ is as follows.  Write the string $x_0x_1\cdots x_{n-3}$. Then immediately before $x_i$ there are $k_i$ larger elements (after that either we run out of numbers or find a smaller one) and immediately after $x_i$ there are $\ell_i$ larger elements (after that either we run out of numbers or find a smaller one).

\begin{defn}\label{def RC}
    Given a non-crossing chord diagram $C$ on $2, 3, \ldots, n, 1$ let $R_C$ be the region corresponding to $C$ according to Conjecture~\ref{conj1}, that is, let $R_C$ be the region in $\mathbb{R}^{n-2}$ with $x_{a-3}=x_{b-3}$ for every chord $\theta_{ab}$ of $C$ and with $x_{a-3}=x_{b-3} < x_{c-3}=x_{d-3}$ whenever chord $\theta_{ab}$ surrounds chord $\theta_{cd}$ in $C$.
\end{defn}

The key to the proof is to understand the behaviour of the explicit formula for $H(x)$ given in \eqref{genH}.  This formula is an alternating sum of terms of the form $\min\{x_a, x_{a+1}, \ldots, x_b\}$.  We refer to each such term as a \emph{window} of length $b-a+1$ and if $x_i$ is the minimum of the $x$s in that window, then we say that window \emph{contributes} $x_i$.  The contribution of $x_i$ to $R(x)$ as a whole is the sum, with signs and factors of $2$ as in \eqref{genH}, of the terms to which it contributes.

\begin{lemma}\label{lem ki li}
  Let $x_0, x_1, \ldots, x_{n-3}$ be totally ordered. With $k_i$ and $\ell_i$ as in Definition~\ref{def ki li}, then the contribution of $x_i$ to $R(x)$ is
  \[
  \begin{cases}
    x_i & \text{if $k_i$, $\ell_i$ are both even} \\
    -x_i & \text{otherwise}
  \end{cases}
  \]
\end{lemma}

\begin{proof}
  The number of length $j$ windows which contribute $x_i$ to $R(x)$ is the number of length $j$ windows beginning at or after $\max\{i-j+1, i-k_i\}$ and ending at or before $\min\{i+j-1, i+\ell_i\}$, which is
  \begin{align*}
    & \min\{i+j-1, i+\ell_i\} - \max\{i-j+1, i-k_i\} - (j-2) \\
    & = \min\{j-1, \ell_i\} - \max\{-j+1, -k_i\} - (j-2) \\
    & = \min\{j-1, \ell_i\} + \min\{j-1, k_i\} - (j-2).
  \end{align*}
  
  Fix $i$ and write $k$ and $\ell$ for $k_i$ and $\ell_i$ respectively, to keep the notation lighter.  Reversing the order if necessary, we can, without loss of generality suppose $k\leq \ell$.  Then we directly compute the contribution of $x_i$ to $R(x)$ to be
  \begin{align*}
    & x_i + 2x_i\sum_{j=2}^{k+1} (-1)^{j-1}j + 2x_i\sum_{j=k+2}^{\ell+1}(-1)^{j-1}(k+1) + 2x_i\sum_{j=\ell+2}^{k+\ell+1}(-1)^{j-1}(k+\ell-j+2) \\
    & = x_i + 2x_i\sum_{j=2}^{k+1} (-1)^{j-1}j + 2x_i(k+1)\sum_{j=2}^{\ell-k+1}(-1)^{j+k-1} + 2x_i\sum_{j=2}^{k+1}(-1)^{j+\ell-1}(k+\ell-j+2) \\
    & = x_i + 2x_i\sum_{j=2}^{k+1}\left((-1)^{j-1} + (-1)^{j+\ell}\right) + 2x_i(k+1)\sum_{j=2}^{\ell-k+1}(-1)^{j+k-1} + 2x_i(k+2)\sum_{j=2}^{k+1}(-1)^{j+\ell-1} \\
    & = \begin{cases}
      x_i + 0 + 0 + 0 & \text{if $\ell$ even and $k$ even} \\
      x_i + 0 + 2x_i(k+1) - 2x_i(k+2) & \text{if $\ell$ even and $k$ odd} \\
      x_i + 2x_i\cdot 2\cdot \frac{k}{2} - 2x_i(k+1) + 0 & \text{if $\ell$ odd and $k$ even} \\
      x_i - 2x_i\cdot 2 \cdot \frac{k+3}{2} + 0 + 2x_i(k+2) & \text{if $\ell$ odd and $k$ odd}
    \end{cases} \\
    & = \begin{cases}
      x_i & \text{if $\ell$ even and $k$ even} \\
      -x_i & \text{if $\ell$ even and $k$ odd} \\
      -x_i & \text{if $\ell$ odd and $k$ even} \\
      -x_i & \text{if $\ell$ odd and $k$ odd}
    \end{cases}
  \end{align*}
  as desired.
\end{proof}

\begin{defn}\label{def bounding cd}
  Given $x_0, x_1, \ldots, x_{n-3} \in \mathbb{R}$ we will call a non-crossing chord diagram on $\{3, 4, \ldots, n\}$ a \emph{bounding chord diagram} for these $x_i$s if there is some total order which refines the order structure of $x_0, x_1, \ldots, x_{n-3}$ as elements of $\mathbb{R}$ (that is, if any $x_i$ are equal, the total order chooses one to be larger and otherwise agrees with the order of the $x_i$ as elements of $\mathbb{R}$) such that, with respect to this total order, the following two properties hold:
  \begin{itemize}
  \item  Suppose $\theta_{ab}$ is a chord and $a<b$.  If $x_{a-3}<x_{b-3}$, then $k_{b-3}$ and $\ell_{b-3}$ are both even while $\ell_{a-3}$ is odd, while if $x_{a-3}>x_{b-3}$, then $k_{a-3}$ and $\ell_{a-3}$ are both even while $k_{b-3}$ is odd.
  \item  Whenever a chord $\theta_{ab}$ surrounds another chord $\theta_{cd}$ then $x_{a-3}, x_{b-3} < x_{c-3}, x_{d-3}$.
  \end{itemize}
\end{defn}
Note that the second condition on a bounding chord diagram is the same as the second condition in the conjecture.  The first condition is designed so that, using the notation of Lemma~\ref{lem ki li}, $k_{b-3}$ and $\ell_{b-3}$ are even and $\ell_{a-3}$ is odd for $x_{a-3}<x_{b-3}$ and analogously for $x_{a-3}>x_{b-3}$.

\begin{lemma}\label{lem bounding cd exists}
For any $x_0, x_1, \ldots, x_{n-3} \in \mathbb{R}$ there is at least one bounding chord diagram and the bounding chord diagrams are determined only by the order structure of $x_0, x_1, \ldots, x_{n-3}$ as elements of $\mathbb{R}$.
\end{lemma}

\begin{proof}
  The definition of a bounding chord diagram only sees the order structure of the $x_i$ as elements of $\mathbb{R}$ so the second part of the lemma is immediate.

  For the first part the proof will be by induction. The base case with one chord is apparent.  For the inductive case, choose a total order extending the order structure of $x_0, \ldots, x_{n-3}$.  Using the notation of Definition~\ref{def ki li}, consider $x_0$. We have $k_0=0$.

\medskip

  \textbf{Case 1:} Suppose first that $\ell_0$ is even.  In this case put a chord from $3$ to $\ell_0+4$.

  Note that $x_{\ell_0+1}< x_0$ in the total order and $k_{\ell_0+1}=\ell_0+1$ which is odd so the first condition holds for this chord.  By construction there are an even number of $x_i$ inside the chord $\theta_{3, \ell_0+4}$ and all these $x_i$ inside this chord have larger values than $x_0$ so the values of $k_i$ and $\ell_i$ are the same whether considered on $x_0, \ldots, x_{n-3}$ or on $x_1, \ldots, x_{\ell_0}$, so we can apply the induction hypothesis to $x_1, \ldots, x_{\ell_0}$ to obtain a bounding chord diagram on $4, \ldots, \ell_0+3$.

  Now consider $x_i$ with $i>\ell_0+1$. The index $\ell_i$ is unchanged whether considered on $0, \ldots, n-3$ or on $\ell_0+2, \ldots, n-3$.  For the index $k_i$, if $x_i>x_{\ell_0+1}$ then $k_i$ is also unchanged between $0, \ldots, n-3$ and $\ell_0+2, \ldots, n-3$ while if $x_i<x_{\ell_0+1}$ then either $k_i$ is unchanged or it is increased by $\ell_0+2$.  In all cases the parity of $k_i$ is unchanged.  Apply the induction hypothesis to $x_{\ell_0+2}, \ldots, x_{n-3}$ to obtain a bounding chord diagram on $\ell_0+5, \ldots, n$.  Since the defining conditions on bounding chord diagrams only involves the $\ell_i$ and $k_i$ via their parity, these chords still satisfy the bounding condition on $x_0, \ldots, x_{n-3}$.

  All together this builds a bounding chord diagram on $x_0, \ldots, x_{n-3}$.

  \medskip

  \textbf{Case 2:} Now suppose $\ell_0$ is odd. Let $i_1, i_2, \ldots, i_m$ be the indices where minima within prefixes of $x_1, \ldots, x_{\ell_0}$ are achieved, that is, $x_{i_j}<x_1, \ldots, x_{i_j-1}$.  Note that $i_1=1$ since the prefix of size 1's minimum is its one element.  Then we have $k_{i_j} = i_j-1$ and $\ell_{i_j} = i_{j+1}-i_j-1$ with the convention that $i_{m+1}=\ell_0+1$.  If all the $\ell_{i_j}$ were odd then all the differences $i_{j+1}-i_j$ would be even and so summing these differences would give $i_{m+1} - i_1 = \ell_0+1-1 = \ell_0$.  But $\ell_0$ is odd, so cannot be the sum of even numbers.  Therefore at least one $\ell_{i_j}$ is even.  Let $j$ be minimal with $\ell_{i_j}$ even.  Then $k_{i_j} = i_j-1 = (i_2-i_1)+(i_3-i_2)+\cdots + (i_j-i_{j-1})$ which is a sum of even numbers by the minimality of $j$ and hence is itself even.   Therefore, we can put in the chord $\theta_{3,i_j+3}$ and this chord will satisfy the conditions of Definition~\ref{def bounding cd}.

  Apply the induction hypothesis to obtain a bounding chord diagram on $x_1, \ldots, x_{i_j-1}$. All the elements strictly between $x_0$ and $x_{i_j}$ are larger than $x_0$ and $x_{i_j}$ so the $k$s and $\ell$s for these elements are unchanged when viewed on $x_0, \ldots x_{n-3}$ and hence these chords continue to satisfy the condition of a bounding chord diagram when viewed on $x_0, \ldots, x_{n-3}$.

  Also apply the induction hypothesis to obtain a bounding chord diagram on $x_{i_j+1}, \ldots,$ $x_{n-3}$.  The $\ell$s for these elements are unchanged when viewed on $x_0, \ldots, x_{n-3}$.  For the $k$s, for elements beyond $x_{\ell_0}$ the parity is unchanged when viewed on $x_0, \ldots, x_{n-3}$ for the same reason as in the case when $\ell_0$ was even.  For elements between $x_{i_j+1}$ and $x_{\ell_0}$ (inclusive) either $k$ is unchanged or is increased by $i_j+1$ which is even as discussed above.  In all cases the parity of the $k$s and $\ell$s is unchanged and so these chords continue to satisfy the bounding condition on $x_0, \ldots, x_{n-3}$.

  This constructs a bounding chord diagram on $x_0, \ldots, x_{n-3}$.
\end{proof}

\begin{lemma}\label{lem no spurious 0s}
  For $x_0, x_1, \ldots, x_{n-3} \in \mathbb{R}$, $R(x)>0$ unless the $x_i$ have a bounding chord diagram $C$ with the additional property that for every chord $\theta_{ab}$, $x_{a-3}=x_{b-3}$, and in this case the point $(x_0, x_1, \ldots, x_{n-3})$ is within the region $R_C$.
\end{lemma}

\begin{proof}
  Take a bounding chord diagram for these $x_i$, which must exist by Lemma~\ref{lem bounding cd exists}.
If there are any equalities among the $x_i$ perturb their values slightly to agree with the total order corresponding to the bounding chord diagram.  For each chord $\theta_{ab}$ with $x_{a-3}<x_{b-3}$ then by Lemma~\ref{lem ki li} the contribution of $x_{a-3}$ and $x_{b-3}$ to $R(x)$ is $x_{b-3}-x_{a-3}>0$ (and analogously if $x_{a-3}>x_{b-3}$ then the contribution is $x_{a-3}-x_{b-3}>0$.  All $x_i$ belong to some chord so $R(x)$ is a sum of contribitions of each chord and these contribitions are all positive.  Now take limits to return to the original values of the $x_i$, which may have equalities.  Since the total order extended the original order structure on the $x_i$, in taking these limits the contribution of each chord remains non-negative and is zero precisely if $x_{a-3}=x_{b-3}$.  Therefore $R(x)>0$ unless $x_{a-3}=x_{b-3}$ for all chords $\theta_{ab}$, in which case $R(x)=0$ and both conditions defining the chord diagram of the conjecture hold, hence $(x_0, x_1, \ldots, x_{n-3})$ is within the region defined by this chord diagram.
\end{proof}

In a noncrossing chord diagram, say two chords are \emph{siblings} if they are both nested directly under the same chord.

\begin{lemma}\label{lem region is 0}
  Every noncrossing chord diagram defines a region where $R(x)=0$.
\end{lemma}

\begin{proof}
  Fix a noncrossing diagram $C$.  Let $x_0, \ldots, x_{n-3}$ be in the region $R_C$. The claim is that the noncrossing diagram is a bounding diagram for $x_0, \ldots, x_{n-3}$.  The second condition is the same between the conjecture and the definition of a bounding diagram, so we only need to check the first condition 

  Choose a total order extending the order structure on $x_0, \ldots, x_{n-3}$
  as follows.  Let $\epsilon$ be smaller than the difference between any two nonequal $x_i$.  If multiple chords of the chord diagram have their $x_i$ equal add $k\epsilon/n$ to the $x_i$s corresponding to both ends of the $k$th chord whose variables had this value from left to right.  Then for each chord of the chord diagram add $\epsilon/(2n)$ to the variable corresponding to its right end point.  By this construction we get a total order such that for each chord $\theta_{ab}$ of the chord diagram with $a-3<b-3$ we have $x_{a-3}<x_{b-3}$ and $x_{b-3}$ is the least element of the total order which is larger than $x_{a-3}$.

  For any chord $\theta_{ab}$ with $a-3<b-3$, consider $k_{a-3}, k_{b-3}, \ell_{a-3}, \ell_{b-3}$.  Note that all chords in a noncrossing chord diagram are of even length.  By the choice of the total order $k_{b-3} = b-a-2$ which is even.  $\ell_{b-3}$ must also be even, since moving rightwards, the first element we encounter which is less than $x_{b-3}$ is either the left endpoint of a sibling of $\theta_{ab}$ (anything nested within a sibling is larger than the left endpoint of the sibling) or is the left endpoint of the chord immediately above $\theta_{ab}$, so the number of elements passed on the way is the entirety of 0 or more siblings and the chords they contain, which is even.  Finally, $\ell_{a-3} = \ell_{b-3}+k_{b-3}+1$ which is odd since we chose our total order so that no elements are between $x_{a-3}$ and $x_{b-3}$.  Therefore $\theta_{ab}$ satisfies the definition of a bounding diagram for $x_0, \ldots, x_{n-3}$, but this holds for all chords, and so this chord diagram is a bounding chord diagram.

  Hence, by the second part of Lemma~\ref{lem no spurious 0s} $R(x)=0$ for this point $x_0, \ldots, x_{n-3}$, but this was an arbitrary point in the region so $R(x)=0$ on the region.
\end{proof}

\begin{proof}[Proof of Conjecture~\ref{conj1}]
By Lemma~\ref{lem no spurious 0s} and Lemma~\ref{lem region is 0}, $R(x)=0$ if and only if $x$ is in one of the regions defined by a noncrossing chord diagram as in the statement of the conjecture.
\end{proof}

\medskip

Note, we can give a direct combinatorial proof that $H(x)$ is invariant under an overall shift of all the variables (a fact that we we also know by physical arguments, but that we did not use above) as follows.  Suppose we shift all variables by adding $y$.  This shifts all minima by $y$ and so, using that $n$ is even,
\begin{align*}
  H(x+y) & = H(x) + \sum_{a=0}^{n-3}y + 2\sum_{a=0}^{n-4}\sum_{b=a+1}^{n-3} (-1)^{b-a}y \\
  & = H(x) + y(n-3)+2y\sum_{a=0}^{n-4}\sum_{c=1}^{n-a-3}(-1)^{c} \\
  & = H(x) + y(n-3) + 2y\sum_{a=0}^{n-4}\begin{cases} -1 & \text{if $n-a-3$ odd} \\0 & \text{if $n-a-3$ even} \end{cases} \\
  & = H(x) + y(n-3) + 2y\sum_{\substack{0\leq a\leq n-4\\a \text{ even}}}-1 \\
  & = H(x) + y(n-3) - 2y\left(\frac{n-4}{2}+1\right) \\
  & = H(x) 
\end{align*}

\section{$\phi^p$ from Triangulations of Extended Diagrams}\label{sectriangp} 
For general $\phi^p$ theories, the global Schwinger formulation presented in \cite{Cachazo:2022voc} is analogous to that of $\phi^4$ but now the regions that compute $A_n^{\phi^p}$ are given by non-crossing $(p-2)$-chord diagrams, see definition below, which are counted by $\textrm{FC}_{(n-2)/(p-2)}(p-2,1)$. Such regions are of dimension $n/(p-2)-1$ in $\mathbb{R}^{n-3}$. The sum over all regions leads to the expected amplitude $A_n^{\phi^p}$, where now each region is connected to a double-ordered cubic amplitude $m_{(n+2(p-3))/(p-2)}(\alpha,\beta)$, as we will comment in section \ref{discs}. We refer the reader to \cite{Cachazo:2022voc} for more details on the global Schwinger formula for $\phi^p$ theories and on the role non-crossing $(p-2)$-chord diagrams play in this formulation. 

In this section, by following the same philosophy as in section \ref{sectriang4}, we perform the study of $\phi^p$ theories. In particular, we will show how we can compute the $\phi^p$ amplitude by triangulating an extended version of non-crossing $(p-2)$-chord diagrams, without ever having to compute the integral of the corresponding global Schwinger formula. The number of triangulated extended diagrams coincides with the number of $\phi^p$ trees, which are counted by $\textrm{FC}_{(n-2)/(p-2)}(p-1,1)$.

To start with, let us recall the definitions of $(p-2)$-chord diagrams and their extension, which in this paper we will call again \textit{extended diagrams} for simplicity. Notice that we will slightly change again the convention on the labelling with respect to \cite{Cachazo:2022voc} for convenience.

\begin{defn}\label{kchords}

Place $n-2$ points labeled $3,4, \ldots ,n$ on the real line in increasing order. A {\it non-crossing $(p-2)$-chord} diagram is a partition of the points into sets of size $p-2$ where each part we call a $(p-2)$-chord, and furthermore, where the $(p-2)$-chords can be drawn on the upper half plane without any crossings. Let us denote the $(p-2)$-chord connecting points $a_1,a_2, \ldots ,a_{p-2}$ as $\theta_{a_1,a_2,...,a_{p-2}}$. We also call $\Theta_{ab}$ to the unique path in a $(p-2)$-chord joining two points $a$ and $b$. 
\end{defn}
	Note, this notion of non-crossing agrees with the usual combinatorial notion of non-crossing set partition, which can be defined rigorously by the condition that there do not exist $a<b<c<d$ with $a,c$ in one part and $b,d$ in another.

\begin{defn}
An extended diagram in this context is a non-crossing $(p-2)$-chord diagram on $n$ points labelled by $\{2,3,4,5,\ldots,n,1\}$ in which $\Theta_{21}$ is always included\footnote{In the definition presented in \cite{Cachazo:2022voc} this chord was a $(p-2)$-chord. In this paper we use a 2-chord to surround all others for simplicity. We will therefore abuse notation since, strictly speaking, now the extended diagram is not a $(p-2)$-chord diagram anymore. One can imagine that we are still talking about a $(p-2)$-chord diagram with an additional $(p-2)$-chord surrounding all others but only the lower part of it, i.e. the path $\Theta_{21}$, is what matters for the argument.}. We also define a meadow of an extended diagram as any region in the diagram delimited by 
more than one chord and by the line where the points lie.  A meadow is an $m$-point meadow if there are $m-1$ $(p-2)$-chords delimiting it. 
\end{defn}

In section \ref{secpcubic} we will comment on the fact that a meadow which is delimited by $m-1$ such paths and the real line corresponds to a biadjoint $m$-subamplitude participating in $m_{(n+2(p-3))/(p-2)}(\alpha ,\beta)$. Moreover, the upper boundary of a meadow, $\Theta_{ab}$, actually corresponds to a propagator of the form $1/X_{a,b}$, with the exception of $\Theta_{21}$. 

Now let us proceed to explain how to triangulate an extended diagram. The rules are analogous to those seen in section \ref{sectriang4} for $\phi^4$, where the triangulating chords are all $2$-chords:

\begin{itemize}
    \item A triangulating chord $\Omega_{ab}$ must start, reading from left to right, from the label $a$ on the right of a lower $(p-2)$-chord of the meadow or on the left of the upper path $\Theta_{ad}$ in the meadow, and must end on the label $b$ on the left of a lower chord of the meadow or on the right of $\Theta_{cb}$.
    \item A triangulating chord must surround at least another $(p-2)$-chord in the diagram.
    \item The triangulation ends when only regions with 4 delimiting chords and lines are left. \end{itemize}
Again these are called triangulations because they will correspond to triangulations in the context of section \ref{secpcubic}.  Once the triangulation of an extended diagram is done, one can easily compute its contribution to the $\phi^p$ amplitude in the following way:
\begin{itemize}
    \item Every already existing path $\Theta_{ab}$ which defines the upper path of a meadow, with the exception of $\Theta_{21}$, and every triangulating chord $\Omega_{ab}$ correspond to a propagator of the form $1/X_{a,b}$.
    \item The contribution of a triangulated diagram to the amplitude is given by multiplying all propagators involved.
\end{itemize}
Now we proceed to show some examples to see how the triangulation works. We start with $\phi^6$ and $n=10$, which involves four extended diagrams. In figure \ref{p6n10} the $4$-chords are shown in black, the different meadows are shaded in colors and the triangulating chords are the red ones.

\begin{figure}[H]
\includegraphics[width=15cm]{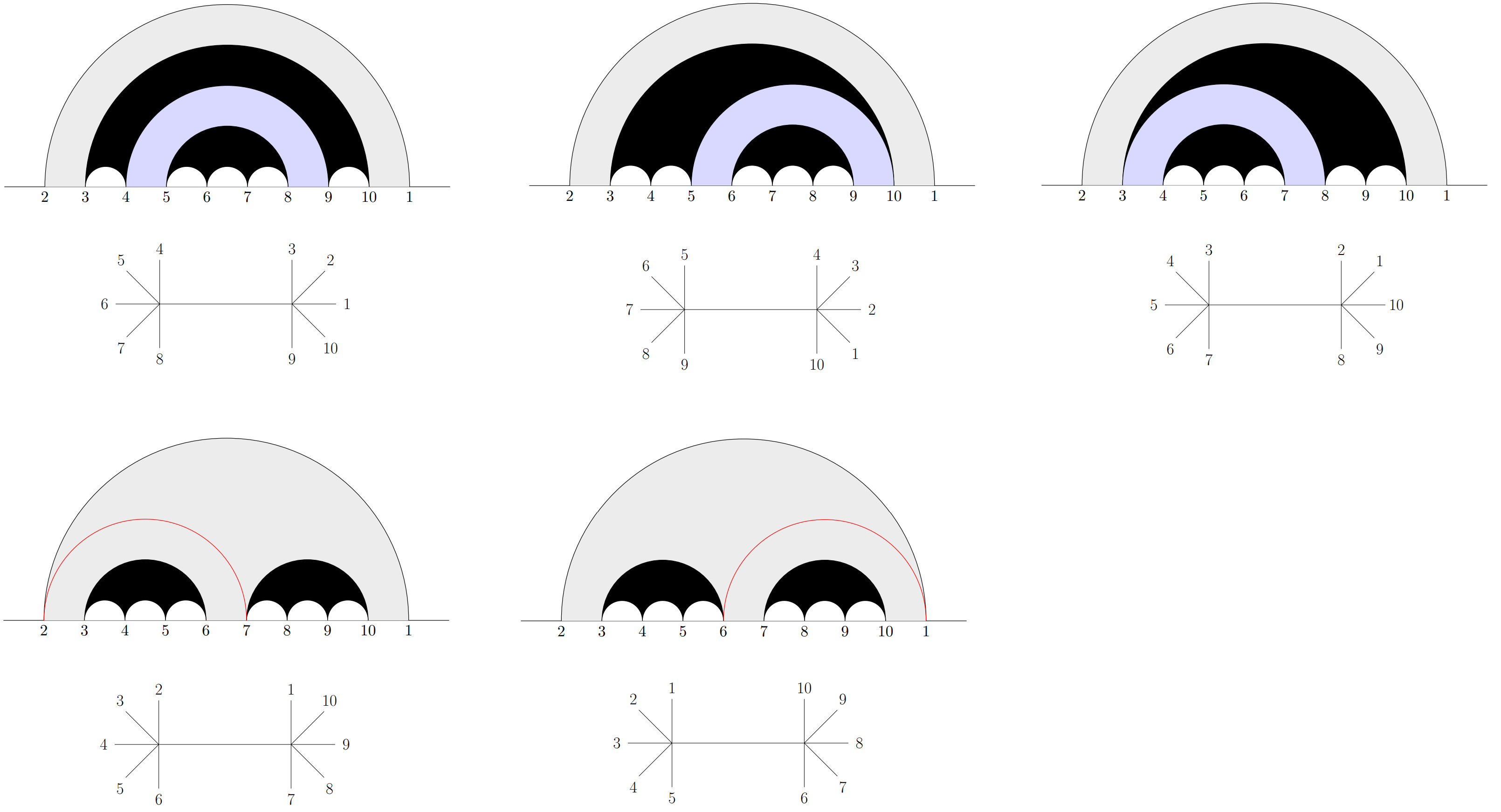}
\centering
\caption{Triangulated extended diagrams for $n=10$ and their corresponding $\phi^6$ tree. The triangulating chords are shown in red, and the colored regions correspond to different meadows. There are four extended diagrams in total, which give rise to five possible triangulations.}
\label{p6n10}
\end{figure}
The first row of extended diagrams in figure \ref{p6n10} are those which only consist of 3-point meadows and are therefore already triangulated. For example, the first extended diagram has two 3-point meadows --shown in gray and blue-- and only one propagator associated to the upper boundary of the blue meadow, i.e. $\Theta_{49}$, which corresponds to a contribution of the form $1/X_{4,9}$, the same as the $\phi^6$ tree below. One can do a similar analysis for the other two diagrams in the first row.

The second row of figure \ref{p6n10} shows the two possible triangulations of the remaining extended diagram, which contains a single 4-point meadow shown in gray. The first diagram in the second row has a triangulating chord of the form $\Omega_{27}$, which gives rise to a propagator of the form $1/X_{2,7}$, the same given by the $\phi^6$ tree below. Similarly, the second diagram has a triangulating chord $\Omega_{61}$, which corresponds to a propagator of the form $1/X_{1,6}$ which is also given by the tree below.

Summing over the five triangulated extended diagrams is equivalent to summing over the five existing $\phi^6$ trees for $n=10$, and one obtains the amplitude $A_{10}^{\phi^6}$.

Now let us see how by triangulating the extended diagrams for $\phi^5$ and $n=11$ leads to the full amplitude. In this case, there are 12 extended diagrams in total. First of all, we show in figure \ref{p5n111} all the 5 extended diagrams which only consist of 3-point meadows and are therefore already fully triangulated.

\begin{figure}[H]
\includegraphics[width=15cm]{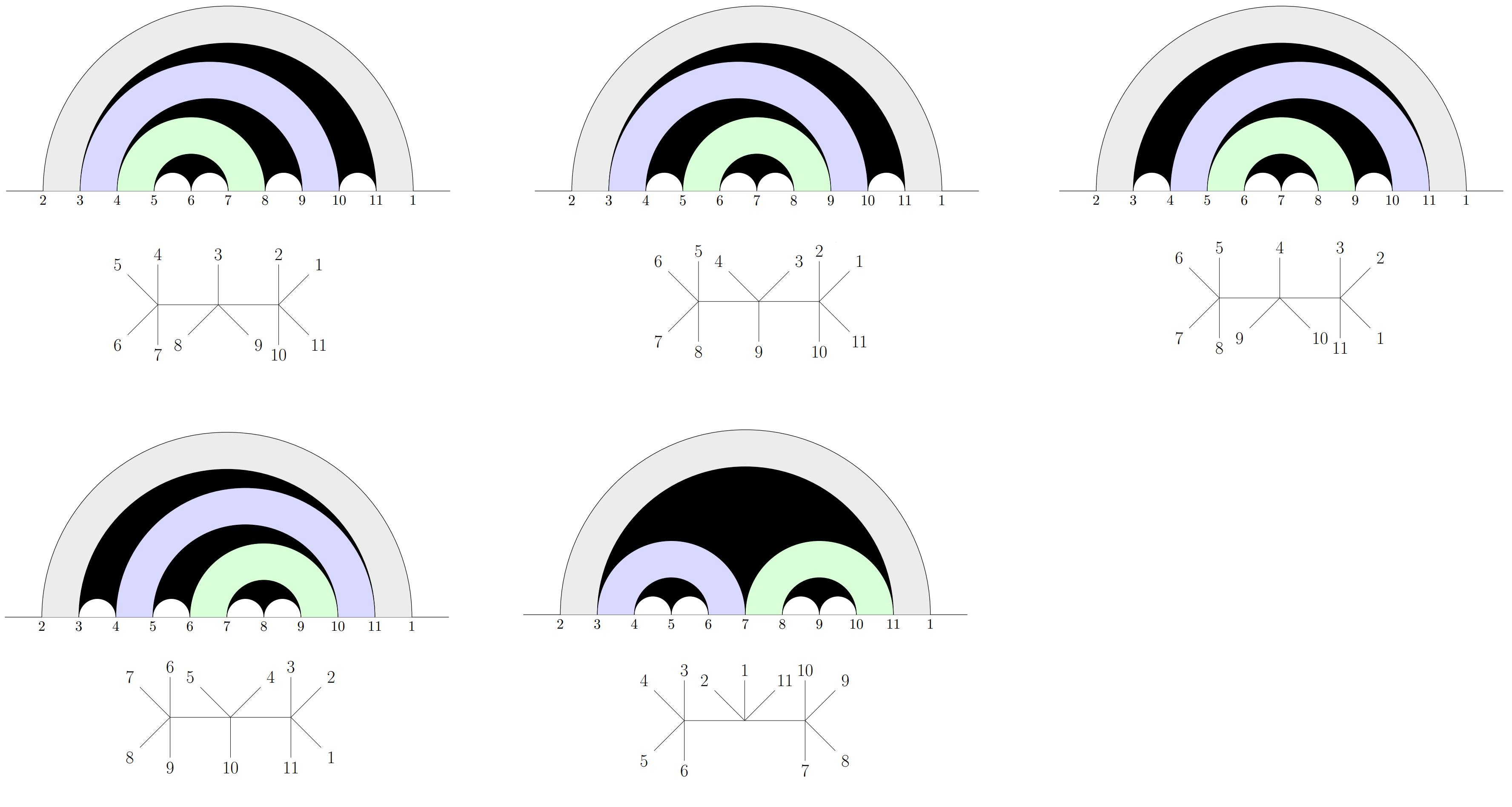}
\centering
\caption{The 5 extended diagrams with only 3-point meadows, colored in different ways, for $n=11$ and the corresponding $\phi^5$ tree shown below.}
\label{p5n111}
\end{figure}
For example, the first diagram in figure \ref{p5n111} has three 3-point meadows with upper boundaries of the form $\Theta_{3,10}$ and $\Theta_{48}$, which correspond to the propagators $1/X_{3,10}$ and $1/X_{4,8}$, the same appearing in the $\phi^5$ tree below. Hence, this diagram contributes
$$\frac{1}{X_{3,10}X_{4,8}}$$
to the amplitude. As another example, the last diagram contains three 3-point meadows with upper boundaries $\Theta_{37}$ and $\Theta_{7,11}$ corresponding to the propagators $1/X_{3,7}$ and $1/X_{7,11}$, which one can also read from the tree below, and it contributes 
$$\frac{1}{X_{3,7}X_{7,11}}$$
to the amplitude.
For the $\phi^5$ and $n=11$ case, we also have 6 extended diagrams with one 3-point meadow and one 4-point meadow, and therefore $6\times2=12$ possible triangulations, as shown in figure \ref{p5n112}.

\begin{figure}[H]
\includegraphics[width=15cm]{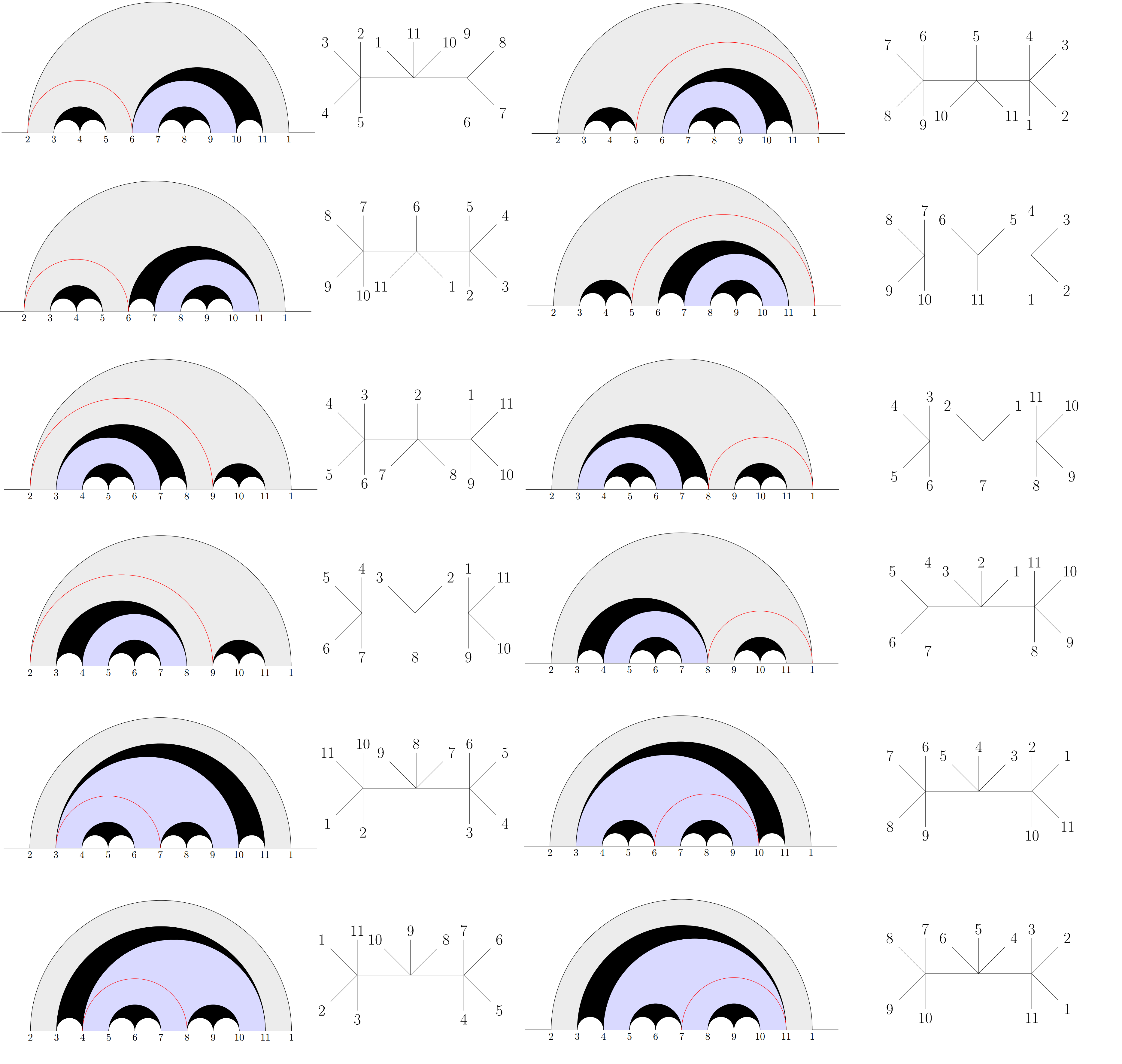}
\centering
\caption{The 6 extended diagrams in $n=11$ with one 3-point meadow and one 4-point meadow. Every row corresponds to an extended diagram, and the two columns correspond to the two possible triangulations, with the corresponding $\phi^5$ tree shown on the right.}
\label{p5n112}
\end{figure}
In this case, every extended diagram has two possible triangulations, due to the fact that each extended diagram contains one 4-point meadow. The extended diagram on the first row in figure \ref{p5n112} has an upper boundary of the form $\Theta_{6,10}$ in the blue meadow, associated to a propagator $1/X_{6,10}$. Each of its triangulations, corresponding to every column in the first row, has a red triangulating chord of the form $\Omega_{26}$ and $\Omega_{51}$, respectively, giving rise to propagators of the form $1/X_{2,6}$ and $1/X_{1,5}$.

Therefore, the contribution of the first extended diagram in figure \ref{p5n112} to the $\phi^5$ amplitude is given by
$$\frac{1}{X_{6,10}}\left(\frac{1}{X_{2,6}}+\frac{1}{X_{1,5}}\right)\,,$$
which is the same given by summing over the two trees on the first row of figure \ref{p5n112}. Again, this expression resembles that of a double-ordered 5-point cubic amplitude $m_5(\alpha,\beta)$, for some orderings $\alpha$ and $\beta$. Finally, one can perform a similar analysis to the rest of the extended diagram and their triangulations.

There is one more extended diagram in $\phi^5$ and $n=11$ that we have not talked about yet, which corresponds to the one containing a single 5-point meadow. We show its 5 possible triangulations in figure \ref{p5n113}.

\begin{figure}[H]
\includegraphics[width=15cm]{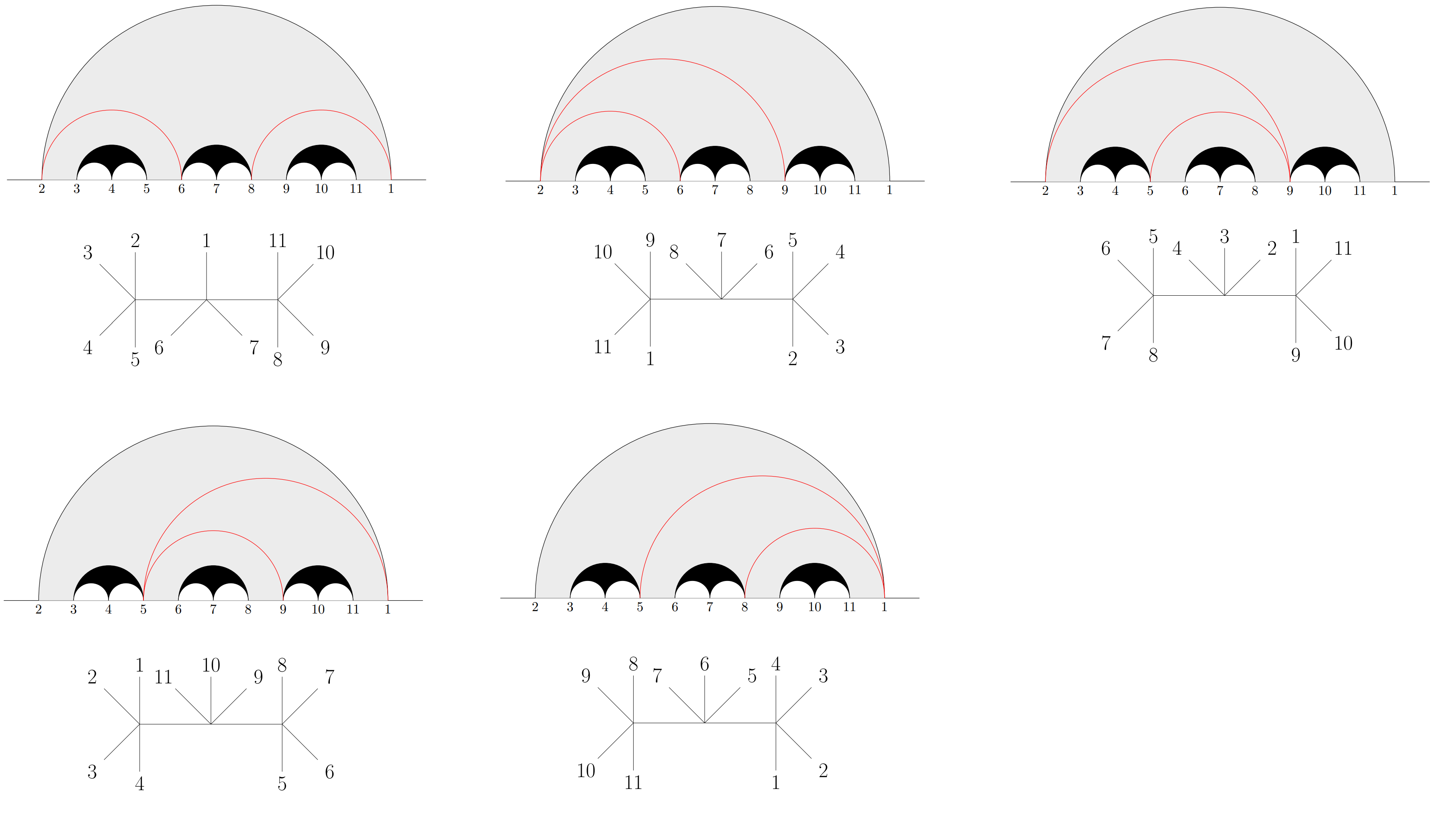}
\centering
\caption{The extended diagram in $n=11$ with one 5-point meadow has 5 possible triangulations, shown together with the corresponding $\phi^5$ tree below.}
\label{p5n113}
\end{figure}
The triangulation corresponding to the left diagram in the first row of figure \ref{p5n113}, has two triangulating chords $\Omega_{26}$ and $\Omega_{81}$. Therefore, its contribution to the $\phi^5$ amplitude is given by
$$\frac{1}{X_{2,6}X_{1,8}}\,,$$
which is the same as that of the $\phi^5$ tree below the triangulated diagram. By a similar analysis done for the rest of the triangulations of this extended diagram one ends up with its contribution to $A_{11}^{\phi^5}$, given by
$$\frac{1}{X_{2,6}X_{1,8}}+\frac{1}{X_{2,9}X_{2,6}}+\frac{1}{X_{2,9}X_{5,9}}+\frac{1}{X_{5,9}X_{1,5}}+\frac{1}{X_{1,5}X_{1,8}}\,,$$
which resembles that of a 5-point diagonal amplitude $m_5(\alpha,\alpha)$. This is not a coincidence, as we will discuss in section \ref{secpcubic}.

To end with this example, one can sum over the contributions of all the FC$_3(3,1)=12$ extended diagrams, which are given by $5\times1+6\times2+1\times5=$FC$_3(4,1)=22$ triangulations --which are in correspondence with the 22 $\phi^5$ trees in $n=11$--, to end up with the expected $A_{11}^{\phi^5}$ amplitude.

\section{Comments on Relation with Cubic Amplitudes}\label{secp4cubic} 
One of the powers of the global Schwinger formula is that it provides a framework to discover some properties of scattering amplitudes which are hard to see from the standard Feynman diagram formulation. For example, in section \ref{sectriang4} we saw how extended diagrams contribute to the $\phi^4$ amplitude in a form which resembles that of a double-ordered cubic amplitude $m_{n/2+1}(\alpha,\beta)$, as first noticed in \cite{Cachazo:2022voc}. In fact, in this work it was proven that, for a given number of particles $n$, extended diagrams consisting of a single $(n/2+1)$-point meadow\footnote{This corresponds to the region in $\textrm{Trop}^+G(2,n)$ coming from the non-crossing chord diagram in which all the chords join two consecutive points, i.e. $\{x_0=x_1,x_2=x_3,...,x_{n-4}=x_{n-3}\}$.} are in bijection with a diagonal amplitude $m_{n/2+1}(\alpha,\alpha)$ for some bijection of the kinematic invariants.

In \cite{Cachazo:2013iea} it was shown that double-ordered cubic amplitudes can be expressed as a product of cubic \textit{subamplitudes}, using a polygon decomposition that we review in the next subsection. Therefore, in \cite{Cachazo:2022voc} it was proposed that $\phi^4$ amplitudes could be expressed as a sum of products of cubic amplitudes, and a formula for the general schematic structure of $A_n^{\phi^4}$ was given based on the Lagrange inversion procedure. For example, using this formula one could write
$$A_8^{\phi^4}=m_5+3m_3m_4\frac{1}{P^2}+m_3^3\left(\frac{1}{P^2}\right)^2\,,$$
where $m_i$ is a cubic subamplitude with $i$ points and $1/P^2$ is a generic propagator. 

Going back to the $n=8$ example in section \ref{sectriang4}, we know that in this case we have an extended diagram consisting of a single 5-point meadow, three extended diagrams consisting of a 3-point meadow together with a 4-point meadow, and one extended diagram consisting of three 3-point meadows, as shown in figure \ref{extp4n8m}. This is precisely the schematic structure of $A_8^{\phi^4}$ shown above, if one relates the $m$-point meadows with the subamplitudes in the expression. Moreover, the number of propagators coincides with the number of chords separating different meadows in the extended diagrams.

In \cite{Cachazo:2022voc} it was also pointed out that extended diagrams for general $\phi^p$ theories are also connected to double-ordered amplitudes $m_{(n+2(p-3))/(p-2)}(\alpha,\beta)$, and a formula for the general schematic structure of $A_n^{\phi^p}$ in terms of products of cubic subamplitudes was also given.

In this section we delve more into the connection between extended diagrams in $\phi^4$ and double-ordered amplitudes. In fact, we propose a combinatorial way to relate any extended diagram with such a schematic structure in lines with the standard polygon decomposition used in double-ordered amplitudes, hence extending the results of \cite{Cachazo:2022voc}. However, we must mention that our new proposed polygon-decomposition does not generate the two permutations $\alpha$ and $\beta$ of the double-ordered amplitude. Instead, it provides an alternative way to understand an extended diagram as a polygon-decomposed object, with different triangulating rules to produce the desired amplitude, which are of course equivalent to the triangulating rules of extended diagrams, hence justifying the term triangulation in the extended diagram case.

\subsection{Review of the Polygon Decomposition of $m_n(\alpha,\beta)$}
In section \ref{sec1} we saw the definition of what double-ordered cubic amplitudes\footnote{In the literature these are  also known as off-diagonal amplitudes.} $m_n(\alpha,\beta)$ are. In this subsection we review a very efficient diagrammatic procedure to compute them, first introduced in \cite{Cachazo:2013iea}, which will turn out to be very convenient to explain the main point of this section.

The procedure is the following. Given $m_n(\alpha,\beta)$, which is defined by any two permutations $\alpha$ and $\beta$, we first draw points on a circle following the permutation $\alpha$. We next join the points according to the ordering defined by $\beta$ inside the circle\footnote{Or viceversa. Notice that $m_n(\alpha,\beta)=m_n(\beta,\alpha)$.}. The resulting object is a set of polygons which correspond to cubic subamplitudes. Each vertex in the polygon is dual to a leg of a tree that would appear after triangulating the polygon (note this is a slightly different notion of duality to that giving the dual quadrangulation used in section~\ref{sec bijection}). Also, the internal vertices joining two subamplitudes define propagators. 

There is also a global sign that multiplies the double-ordered amplitude. We can obtain the sign e.g. by following the convention of \cite{Mizera:2016jhj}, which consists of computing the relative winding number $w_{\alpha,\beta}$ between the two permutations. The overall sign is given by $(-1)^{1+w_{\alpha,\beta}}$.

\begin{figure}[H]
\includegraphics[width=5cm]{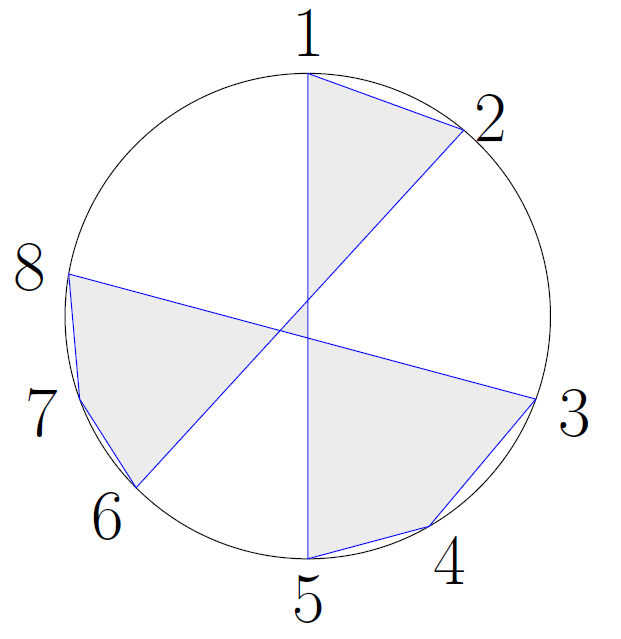}
\centering
\caption{The permutations $\alpha=(12678345)$ in blue and $\beta=\mathbb{I}=(12345678)$ in black.}
\label{off}
\end{figure}

It is instructive to see an example of how the whole procedure works. Consider, e.g., the double-ordered amplitude $m_8(12678345,\mathbb{I})$. The two permutations are pictorially represented in figure \ref{off}. Notice from this figure that we can see four different polygons whose interior has been shaded, i.e. two triangles and two squares, and three internal vertices. This means that this double-ordered amplitude has a schematic form
$$m_8(12678345,\mathbb{I})=m_3^2m_4^2\left(\frac{1}{P^2}\right)^3\,.$$
The square formed by vertices 3, 4, 5 and an internal one corresponds to a 4-point subamplitude given by
$$\frac{1}{s_{34}}+\frac{1}{s_{45}}\,.$$
Similarly, the square formed by vertices 6, 7, 8 and another internal one also corresponds to a 4-point subamplitude, in this case
$$\frac{1}{s_{67}}+\frac{1}{s_{78}}\,.$$
The two triangles correspond to 3-point subamplitudes, which are constants that we take as 1 by convention. There are also 3 internal vertices, which are propagators of the form $1/s_{12}$, $1/X_{3,6}$ and $1/X_{1,6}$. Moreover, the relative winding number in this case is given by $w_{\alpha,\beta}=2$ and the amplitude is, therefore,
$$m_8(12678345,\mathbb{I})=-\frac{1}{s_{12}X_{3,6}X_{1,6}}\left(\frac{1}{s_{34}}+\frac{1}{s_{45}}\right)\left(\frac{1}{s_{67}}+\frac{1}{s_{78}}\right)\,.$$

\subsection{Polygon Decomposition of a $\phi^4$ Extended Diagram}
In this subsection we start by proposing a polygon-decomposition like procedure for an extended diagram participating in $\phi^4$ theories, motivated by the nested structure that extended diagrams have in terms of meadows. The rules to express a generic extended diagram into a polygon-decomposed object are the following:

\begin{itemize}
    \item Given an extended diagram, we draw an $m$-gon for every $m$-point meadow in the diagram.
    \item For every $m$-gon, two of its vertices are given by the two labels of the upper chord defining the meadow where it comes from. The rest of the vertices are labelled by the pair of points joined by lower chords $\theta_{ab}$ defining the meadow, in the form $[a,b]:=a,a+1,\cdots,b-1,b$.
    \item Every propagator $\theta_{ab}$ separating two meadows in the extended diagram corresponds to an internal vertex, which is labelled $b$ and lies inside the polygon corresponding to the meadow surrounding the other one from above --with $b>a$-- and which joins the two corresponding $m$-gons. These kind of vertices lie strictly inside a circle. The rest of the vertices lie on the boundary of the circle.
\end{itemize}
Let us see some examples to help us clarify the procedure. Recall that for $n=6$ we have two different extended diagrams, which are shown in figure \ref{extp4n6m}. The proposed polygon decomposition is shown in figure \ref{polp4n6}.

\begin{figure}[H]
\includegraphics[width=15cm]{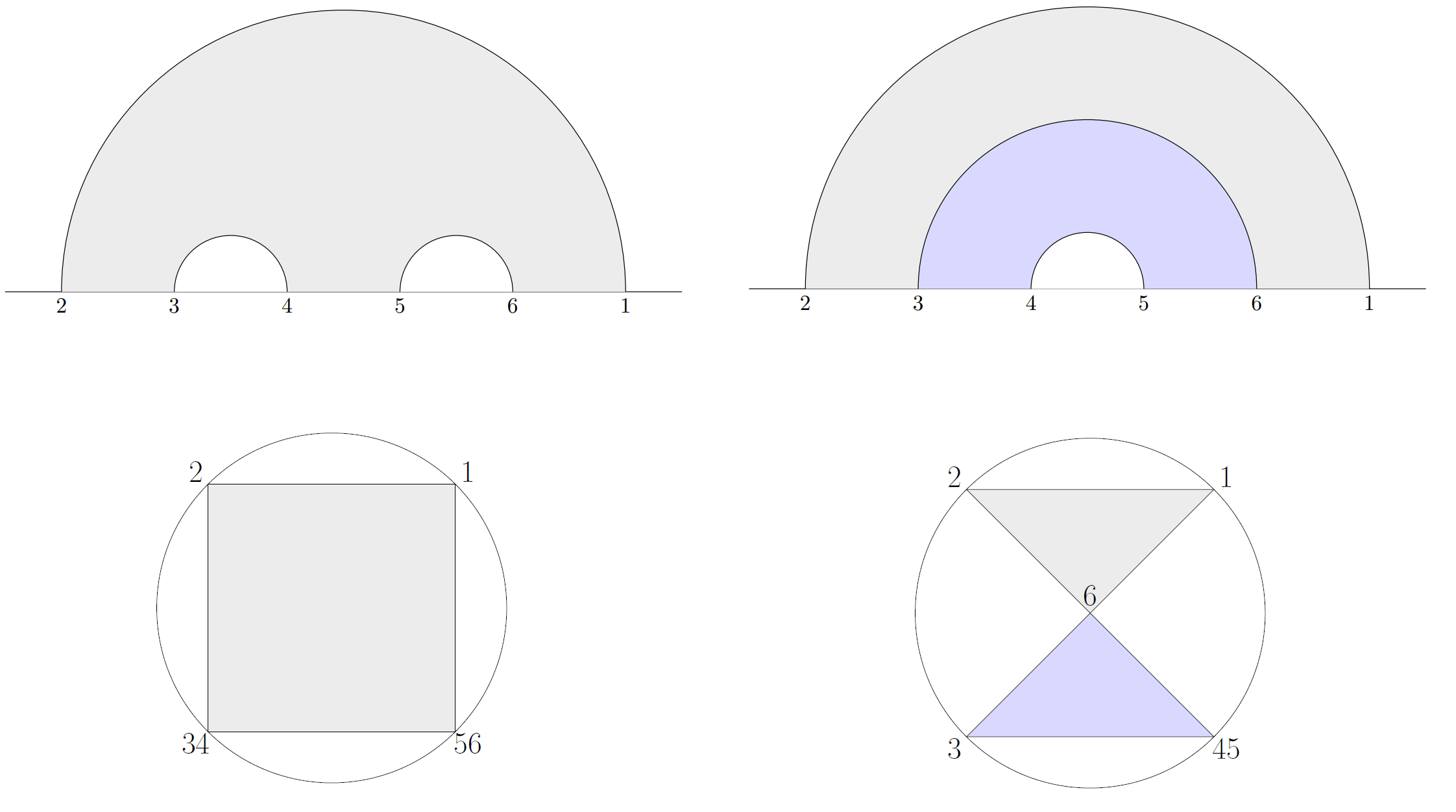}
\centering
\caption{Polygon decomposition of the extended diagrams in $n=6$.}
\label{polp4n6}
\end{figure}
The extended diagram on the left consists of a single $4$-point meadow, so its polygon decomposition corresponds to a square, as shown below. The extended diagram on the right consists of two 3-point meadows, shaded in gray and blue, which give rise to the polygon decomposition shown below and which corresponds to two triangles. Note that the triangles are joined by an internal vertex, labelled by 6, which is the greater label in the propagator $\theta_{36}$ in the extended diagram.

\begin{figure}[H]
\includegraphics[width=10cm]{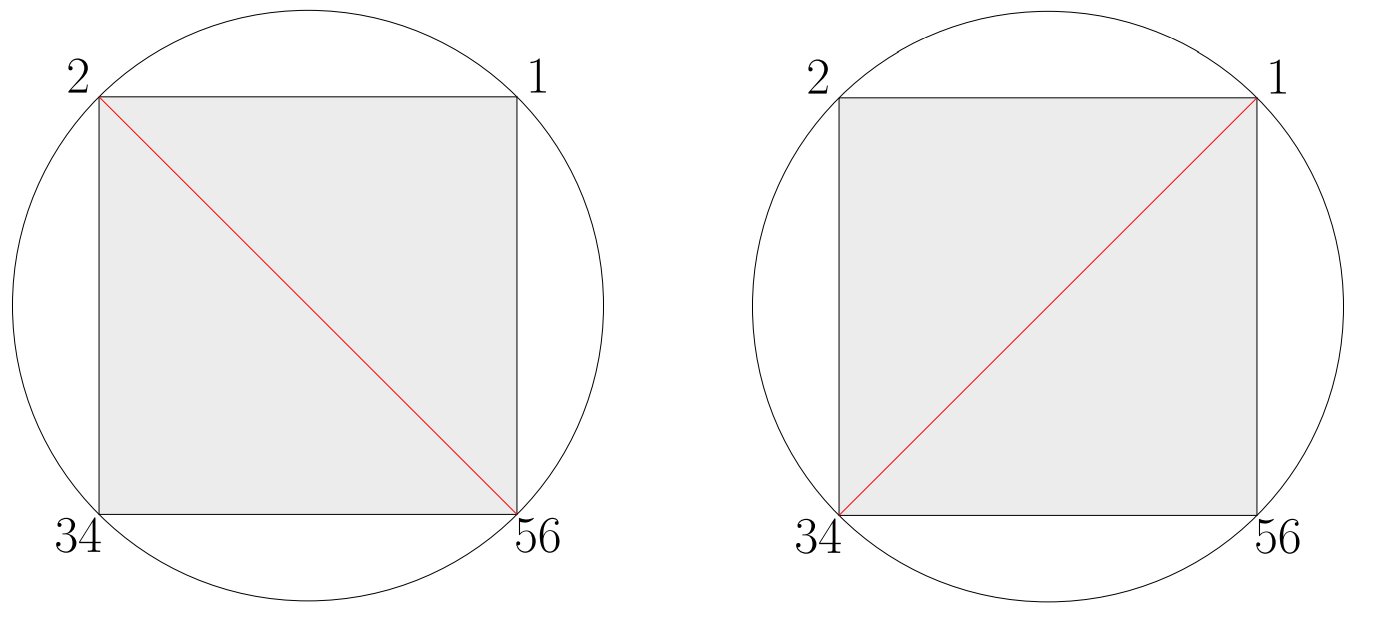}
\centering
\caption{Triangulations of a polygon-decomposed extended diagram in $n=6$.}
\label{trin61pol}
\end{figure}

One might be tempted to say that the extended diagram on the left corresponds to a diagonal amplitude $m_{{4}}(\alpha,\alpha)$ with $\alpha=({1},{2},{34},{56})$. Similarly, the diagram on the right might correspond to an off-diagonal amplitude of the form $m_{{4}}(\alpha,\beta)$, with the permutations defined by $\alpha=({1},{2},{3},{45})$ and $\beta=({1},{2},{45},{3})$. However, the $m_{{4}}(\alpha,\beta)$ amplitudes in general cannot be read in a standard way from the permutations presented in this picture, since there is a caveat when it comes to triangulating the polygons. 

Notice that the two possible triangulations of the extended diagram on the left, which are those in figure \ref{trin61}, are equivalent to the triangulations of the square, shown in figure \ref{trin61pol}. The triangulation on the left would give us the expected contribution of $1/s_{{2},{34}}=1/s_{{56},{1}}$. But the triangulation on the right, if not interpreted carefully, would give us $1/s_{{34},{56}}=1/s_{{1},{2}}$. This is obviously not what we want, and the reason is that the corresponding triangulating chord in the extended diagram started from label 4 and did not include label 3. 

\begin{figure}[H]
\includegraphics[width=15cm]{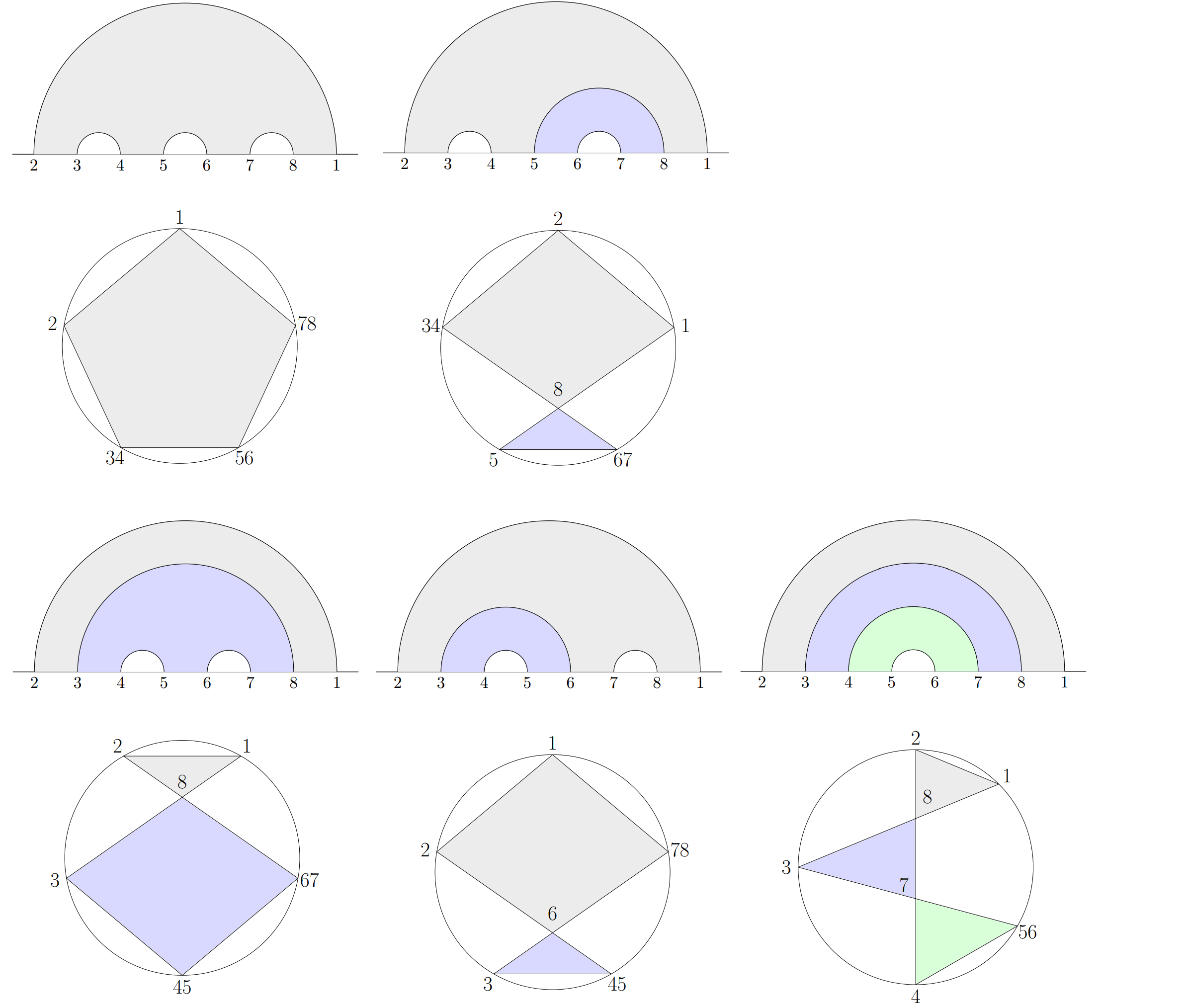}
\centering
\caption{Extended diagrams in $n=8$ and their polygon decomposition shown below.}
\label{p4n8pol}
\end{figure}

This motivates us to introduce a couple more new rules in order to read a triangulation of a polygon-decomposed extended diagram:

\begin{itemize}
    \item We say that a triangulating chord $\Omega_{a,b}$ inside a polygon starts at a vertex $a$ and ends at a vertex $b$ when $a<b$ in the ordering $(2,3,4,...,n,1)$. If a triangulating chord inside a polygon starts at a vertex labelled by a pair of points $ij$, it partitions the pair and label $j$ goes along with the following vertex in the polygon. 
    \item If an $m$-gon with vertices $\{a_1,a_2,...,a_i,...a_m\}$ has an internal vertex $a_i$, the propagator corresponding to that vertex is $1/X_{{a}_1,{a_m+1}}$. It is also read as $1/X_{{b}_1,{b_{l}+1}}$, where $\{b_1,b_2,...,b_l\}$ are the remaining vertices of the polygon that shares the same internal vertex as the previous one, with $a_i\notin\{b_1,b_2,...,b_l\}$, and which corresponds to a meadow below the chord in the extended diagram containing $a_i$.
\end{itemize}
Therefore, now the triangulation on the left of figure \ref{trin61pol} reads as $1/s_{{2},{34}}=1/s_{{56},{1}}$ and the one on the rights reads as $1/s_{{4},{56}}$, as we expected. Moreover, the already triangulated diagram in figure \ref{polp4n6} contributes as $1/s_{{3},{45}}$.

Now let us consider the more interesting case of $n=8$. In figure \ref{p4n8pol} we show the polygon decomposition of the five extended diagrams previously shown in figure \ref{extp4n8m}, where we have colored the meadows and their corresponding n-gons in the same way to make the correspondence more clear. The extended diagram on the bottom-right only consists of 3-point meadows, which are given by triangles in its polygon decomposition. Therefore, it is already triangulated. Its contribution to the amplitude $A_8^{\phi^4}$ is then given by the product of its two internal vertices or propagators, which are of the form $1/X_{3,8}$ and $1/X_{4,7}$.

In figure \ref{trin81pol} we show the five triangulations of the pentagon, which correspond to the five triangulations of the extended diagram shown in figure \ref{trin81}, and presented in the same order.

\begin{figure}[H]
\includegraphics[width=15cm]{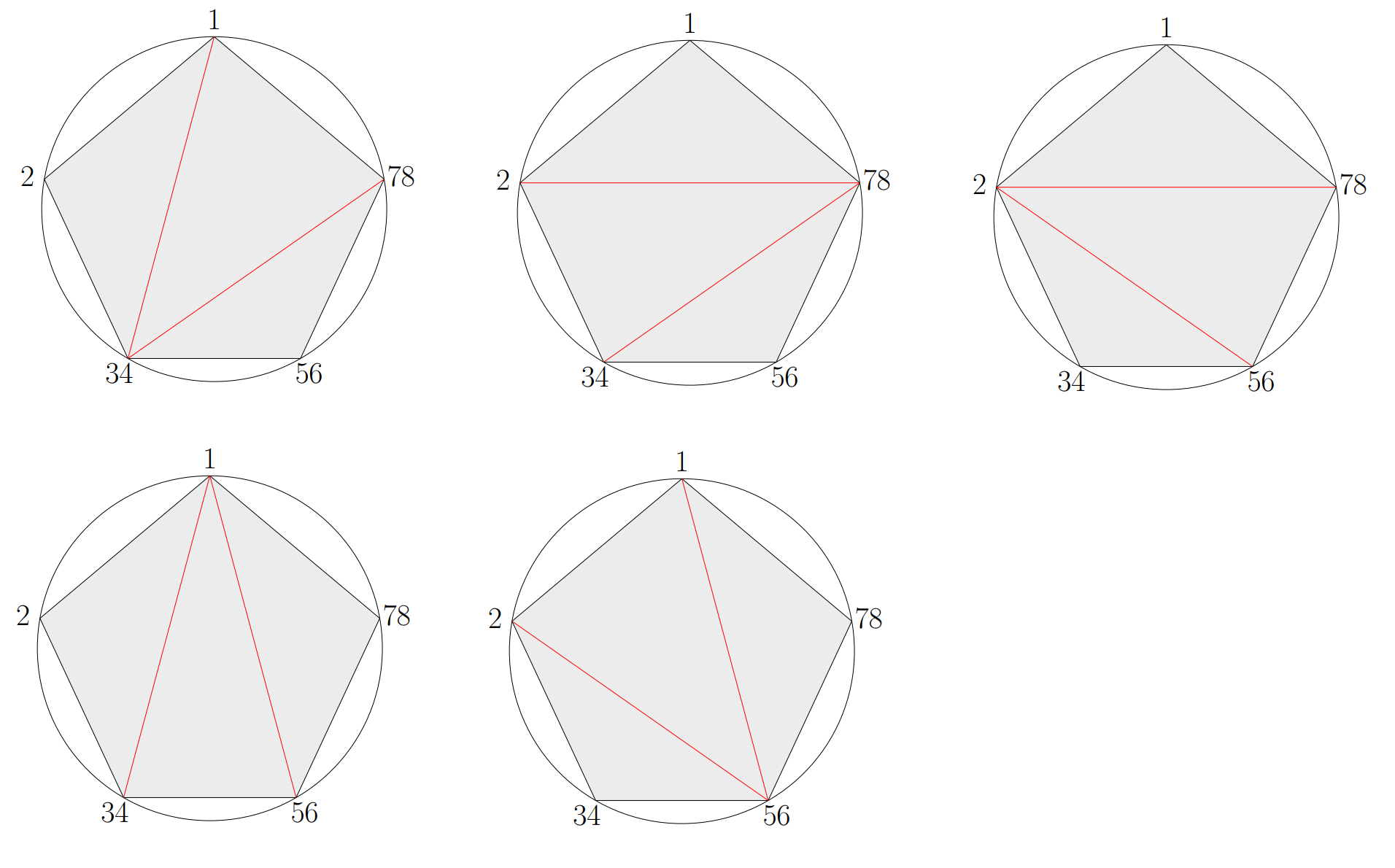}
\centering
\caption{Triangulations of the pentagon corresponding to the extended diagram with a single 5-point meadow in $n=8$. They are the same as those in figure \ref{trin81}.}
\label{trin81pol}
\end{figure}
Going from left to right, top to bottom, the first triangulation gives a contribution of the form $1/X_{4,7}\times1/X_{1,4}$. This is because the triangulating chords $\Omega_{34,78}$ and $\Omega_{34,1}$ start from the pair $34$, thus splitting it into 3 and 4.
Similarly, the rest of the triangulations in figure \ref{trin81pol} give contributions of the form
$$\frac{1}{X_{2,7}X_{4,7}}\,, \hspace{2mm} \frac{1}{X_{2,7}X_{2,5}}\,, \hspace{2mm} \frac{1}{X_{1,4}X_{1,6}}\,, \hspace{2mm} \frac{1}{X_{2,5}X_{1,6}}\,,$$
respectively. Thus summing over the five triangulations of the pentagon we obtain the contribution $A_8^{\phi^4:(1)}$ seen in \eqref{qmet}.

\begin{figure}[H]
\includegraphics[width=15cm]{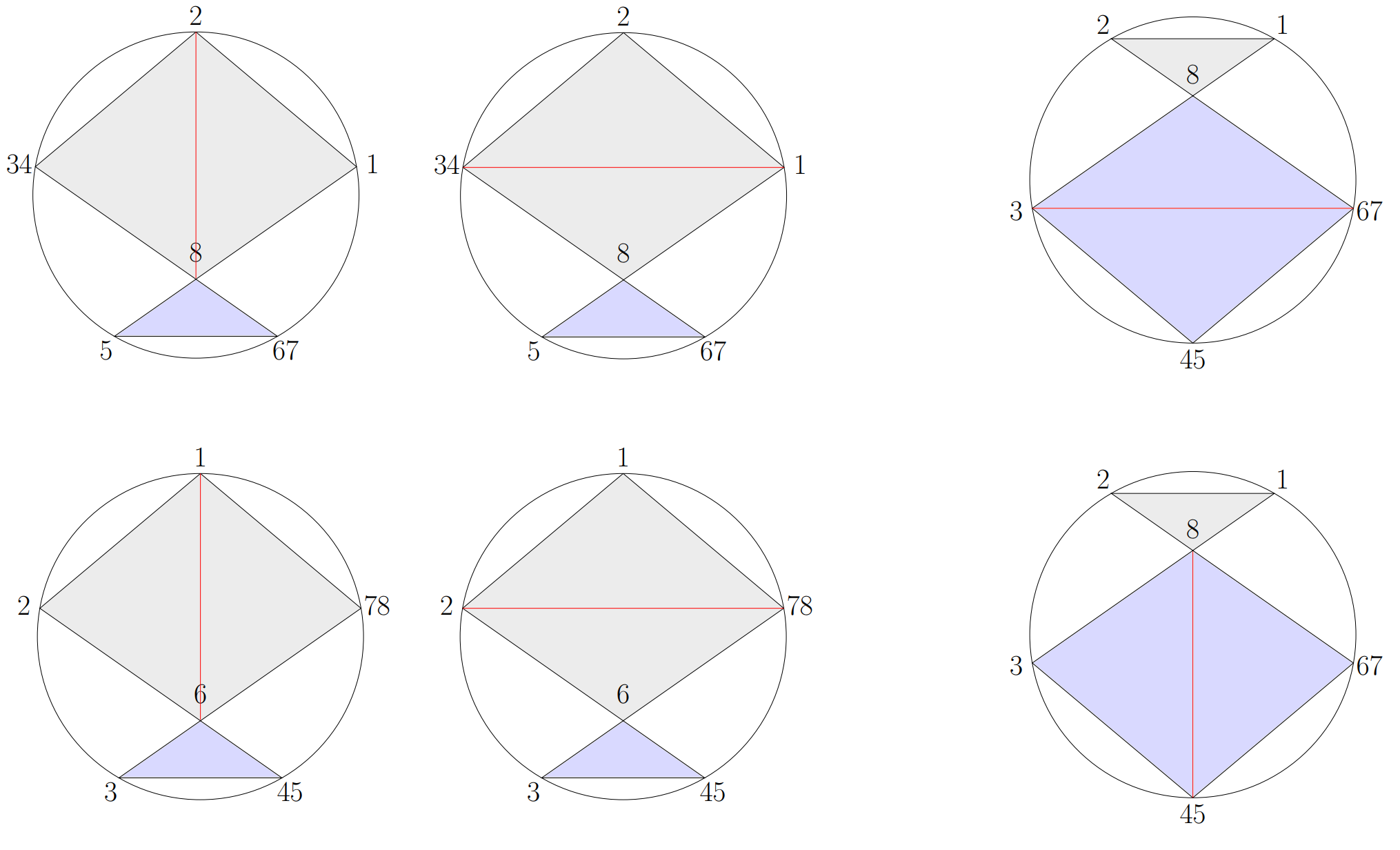}
\centering
\caption{Rest of the triangulations of the polygon-decomposed extended diagrams in $n=8$. They are the same as those in figure \ref{trin82}.}
\label{trin82pol}
\end{figure}

Figure \ref{trin82pol} shows the triangulations for the rest of the polygon-decomposed extended diagrams that appear in $n=8$. The first two triangulations in the first row are those of the same polygon decomposition, and their sum gives a contribution of 
$$\frac{1}{X_{5,8}}\left(\frac{1}{X_{2,5}}+\frac{1}{X_{1,4}}\right)\,,$$
which agrees with that of $A_8^{\phi^4:(2)}$ seen in \eqref{qmet}.

Similarly, by summing over the first two triangulations in the second row of figure \ref{trin82pol}, which are those of the same polygon decomposition, we obtain a contribution of 
$$\frac{1}{X_{3,6}}\left(\frac{1}{X_{1,6}}+\frac{1}{X_{2,7}}\right)\,,$$
which agrees with that of $A_8^{\phi^4:(4)}$ seen in \eqref{qmet}.

Finally, the two triangulations in the third column give a contribution of 
$$\frac{1}{X_{3,8}}\left(\frac{1}{X_{3,6}}+\frac{1}{X_{5,8}}\right)\,,$$
which agrees with that of $A_8^{\phi^4:(3)}$ in \eqref{qmet}.

Therefore, the sum of all the triangulations of the polygon-decomposed diagrams in $n=8$ compute the full amplitude $A_8^{\phi^4}$.

To finish with the section, we will have a look at two triangulations of a more complicated polygon decomposition. In fact, we will study that of the extended diagram appearing in figure \ref{p4n16}, now represented in figure \ref{p4n16pol} on the left with the polygons in same color as their corresponding meadows.

\begin{figure}[H]
\includegraphics[width=15cm]{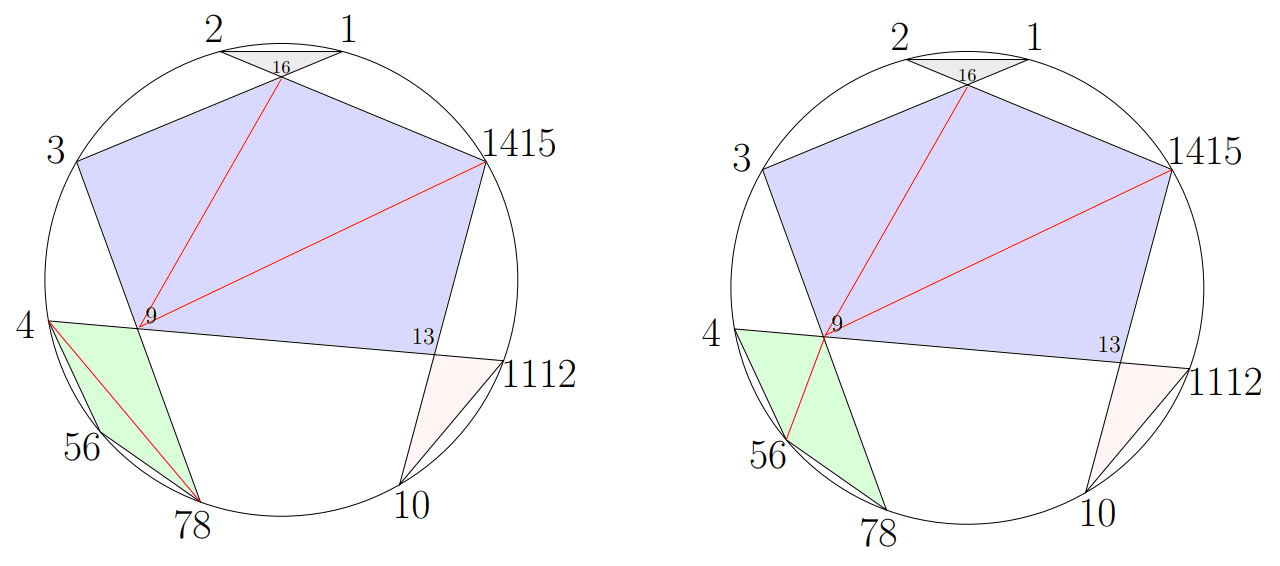}
\centering
\caption{Triangulations of a polygon decomposition in $n=16$. The triangulation on the left is equivalent to that in figure \ref{p4n16}.}
\label{p4n16pol}
\end{figure}
In the case on the left of figure \ref{p4n16pol}, we have three internal vertices which correspond to propagators. The vertex joining the triangle $\{16,1,2\}$ and the pentagon gives a propagator of the form $1/X_{3,16}$. Likewise, the vertex joining the pentagon and the triangle with external vertices $10$ and $11,12$ corresponds to a propagator of the form $1/X_{10,13}$. Finally, the vertex joining the pentagon and the square gives rise to $1/X_{4,9}$.

In this diagram, we have three triangulating chords, shown in red. The triangulating chord $\Omega_{4,78}$ gives rise to $1/X_{4,7}$. The chord $\Omega_{9,1415}$ generates $1/X_{9,14}$ and the chord $\Omega_{9,16}$ generates $1/X_{9,16}$. By multiplying all the propagators one ends up with the same result as that seen from figure \ref{p4n16}.

In the case on the right of figure \ref{p4n16pol}, the only difference is that now the triangulation of the square, i.e. $\Omega_{56,9}$, gives a contribution of $1/X_{6,9}$, according to our triangulating rules.

Hence, in this subsection we have provided a polygon-decomposition pictorial procedure to connect extended diagrams with products of cubic subamplitudes, which slightly differs from the standard polygon decomposition for off-diagonal amplitudes. Of course, this procedure is exactly equivalent to the triangulations of extended diagrams. The main problem is that it is non-obvious how $\phi^4$ trees coming from a triangulation of an extended diagram can be read as $\phi^3$ trees contributing to the off-diagonal amplitude. One of the main obstacles we encounter in this regard is that one has to be careful in that the variable $x_{a-3}$ coming from the
parameterization \eqref{preT}, which we relate with the label $a$ on the real line in extended diagrams, does not exactly correspond to a single particle $a$, as this variable appears in all rows $r\geq a$. However, due to the similarities with the polygonal decomposition of double-oredered amplitudes, we hope this approach might provide a step further into finding the permutations $\alpha$ and $\beta$ of the expected amplitude from an extended diagram.

\section{Discussions}\label{discs}
In this note we have shown that we can obtain tree-level $\phi^4$ scattering amplitudes by triangulating extended diagrams. We have proved that these diagrams are in bijection with the regions in $\textrm{Trop}^+G(2,n)$ contributing to the global Schwinger formula for $A_n^{\phi^4}$. This means that every extended diagram contains at least one $\phi^4$ tree. The reason is that, in general, an extended diagram by itself does not account for all the propagators of a tree, but needs to be triangulated. In fact, we have also proven that every triangulated extended diagram is in bijection with a $\phi^4$ Feynman diagram, and that all the triangulated extended diagrams add up to $A_n^{\phi^4}$. This implies that every triangulating chord in an extended diagram is a quadrangulating chord in an even n-gon. However, to account for all the quadrangulating chords of the n-gon, i.e. the propagators of the tree, one must also consider the chords appearing in the original non-triangulated extended diagram. So in a sense, the collection of extended diagrams, each of which characterizes a polyhedral cone in $\textrm{Trop}^+G(2,n)$ which is related to a double-ordered cubic amplitude, partitions the $\phi^4$ amplitude in a new way.

Moreover, we have also seen how triangulations of the extended diagrams that participate in the global Schwinger formulation of $A_n^{\phi^p}$ also lead to $\phi^p$ amplitudes, although we lack of a formal proof for the general case.

An interesting property of this procedure is that extended diagrams are very easily obtained by using a simple combinatorial operation, and so are their triangulations. The procedure to compute $\phi^p$ amplitudes presented in this work actually differs from the Stokes polytopes and accordiohedra literature, as well as other approaches  \cite{Banerjee:2018tun,Salvatori:2019phs,Aneesh:2019cvt,Srivastava:2020dly,Raman:2019utu,Aneesh:2019ddi,Kojima:2020tox,Kalyanapuram:2020vil,John:2020jww,Kalyanapuram:2020axt,Jagadale:2021iab,Baadsgaard:2015ifa}, as pointed out in \cite{Cachazo:2022voc}, since each of our regions/extended diagrams is isomorphic to an associahedron or to an intersection of two associahedra with a numerical weight of 1 associated to them. It would be very interesting to make a connection between the two approaches in the future.

It is important to mention that a global Schwinger formula for general positive tropical Grassmannians $\textrm{Trop}^+G(k,n)$ also exists \cite{Cachazo:2020wgu}, which is naturally related to generalized Grassmannian amplitudes, also known as CEGM amplitudes \cite{Cachazo:2019ngv}. This framework provides an intriguing generalization of the notion of quantum field theory, with many exciting results so far, see e.g. \cite{Early:2021tce,Arkani-Hamed:2019mrd,Cachazo:2022vuo,GarciaSepulveda:2019jxn,Early:2023dlt}. For example, the analogous objects to Feynman diagrams that compute CEGM amplitudes consist of arrays of lower-point Feynman diagrams \cite{Borges:2019csl,Cachazo:2019xjx,Cachazo:2022pnx,Cachazo:2023ltw}. It would be interesting to explore the possibility of finding what the $k>2$ analogous to $\phi^p$ amplitudes are as sums of $\phi^p$-like arrays of Feynman diagrams, that could be obtained by triangulating a more general version of extended diagrams.

Moreover, the global Schwinger formula was recently generalized to higher loop order in \cite{Arkani-Hamed:2023lbd}. All these frameworks leave plenty of room for exciting future research directions.

To finish with the paper, we comment on a few interesting research lines, more related to the work presented here.

\subsection{Prospects for Extending Proofs in $\phi^p$}

It would be interesting to have a proof of how triangulations of the extended diagrams that participate in the global Schwinger formulation of $A_n^{\phi^p}$ also lead to $\phi^p$ amplitudes, analagously to the proof of section~\ref{sec bijection} in the $\phi^4$ case.  Hopefully similar ideas would succeed, but there are subtleties on account of the interplay between the $(p-2)$-chords and the $2$-chords.

Furthermore, in section \ref{proof} we gave a proof of Conjecture \ref{conj1}, based on the expression \eqref{genH} for $H(x)$ in $\phi^4$. Namely, we proved that the solutions to $H(x)=0$ are in bijection with the regions defined by non-crossing chord diagrams that contribute to the amplitude $A_n^{\phi^4}$. An analogous conjecture was proposed in \cite{Cachazo:2022voc} for general $\phi^p$ theories. In particular, Conjecture 7.3 of \cite{Cachazo:2022voc} proposes a similar bijection between regions contributing to $A^{\phi^p}_n$ or, equivalently, regions where $H^{\phi^p}(x)=0$ and non-crossing $(p-2)$-chord diagrams. However, no such function $H^{\phi^p}(x)$ was given. 

Here we conjecture that, in fact, the function whose zeros give the desired regions is

$$H^{\phi^p}(x)=\sum_{a=0}^{n-3}x_a-\sum_{{\cal X}}\textrm{min}(x_a,x_{a+1},...,x_b)+2\sum_{{\cal Y}}\textrm{min}(x_a,x_{a+1},...,x_b)\,,$$
where ${\cal X}$ means summing for $a<b$ over ordered pairs $(a,b)$ such that $b-a\in\{1,p-3,p-1,2p-5,2p-3,3p-7,3p-5,...,n-3\}$. The domain ${\cal Y}$ means summing for $a<b$ over ordered pairs $(a,b)$ such that $b-a\equiv0$ mod $p-2$, up to $n-3$.


Because of the very similar form of the conjectured form of $H^{\phi^p}$, one would hope that similar methods, like those in section \ref{proof}, would work to prove Conjecture 7.3 of \cite{Cachazo:2022voc}.

\subsection{Polygon Decomposition of an Extended Diagram in $\phi^p$} \label{secpcubic} 
Of course, one of the most pressing problems is to be able to figure out the precise bijection between an extended diagram and a double-ordered cubic amplitude. We have made a small step towards this goal in section \ref{secp4cubic}, by proposing a way to read the polygonal decomposition of an extended diagram in $\phi^4$. However, we still lack of a precise combinatorial method that would give the desired permutations from it.

In this subsection we start the exploration of the polygon decomposition of extended diagrams for general $\phi^p$ theories.

For example, let us start with $p=6$ and $n=10$. Figure \ref{trinp6n10pol} shows the already triangulated extended diagrams, together with their corresponding proposed polygon decomposition. Note from the three top diagrams in this example that only one of the labels in the 4-chord separating the two meadows becomes an external vertex. Namely, the label on the left of the upper chord of the lower meadow. The rest of the labels become internal vertices inside the upper polygon. One difference from the $\phi^4$ case is that now we must consider internal labels on the left and on the right, see e.g. the first polygon decomposition. Moreover, by looking at the two bottom diagrams in the figure, we see again how a triangulating chord only picks the higher label between those where it starts from. 

\begin{figure}[H]
\includegraphics[width=15cm]{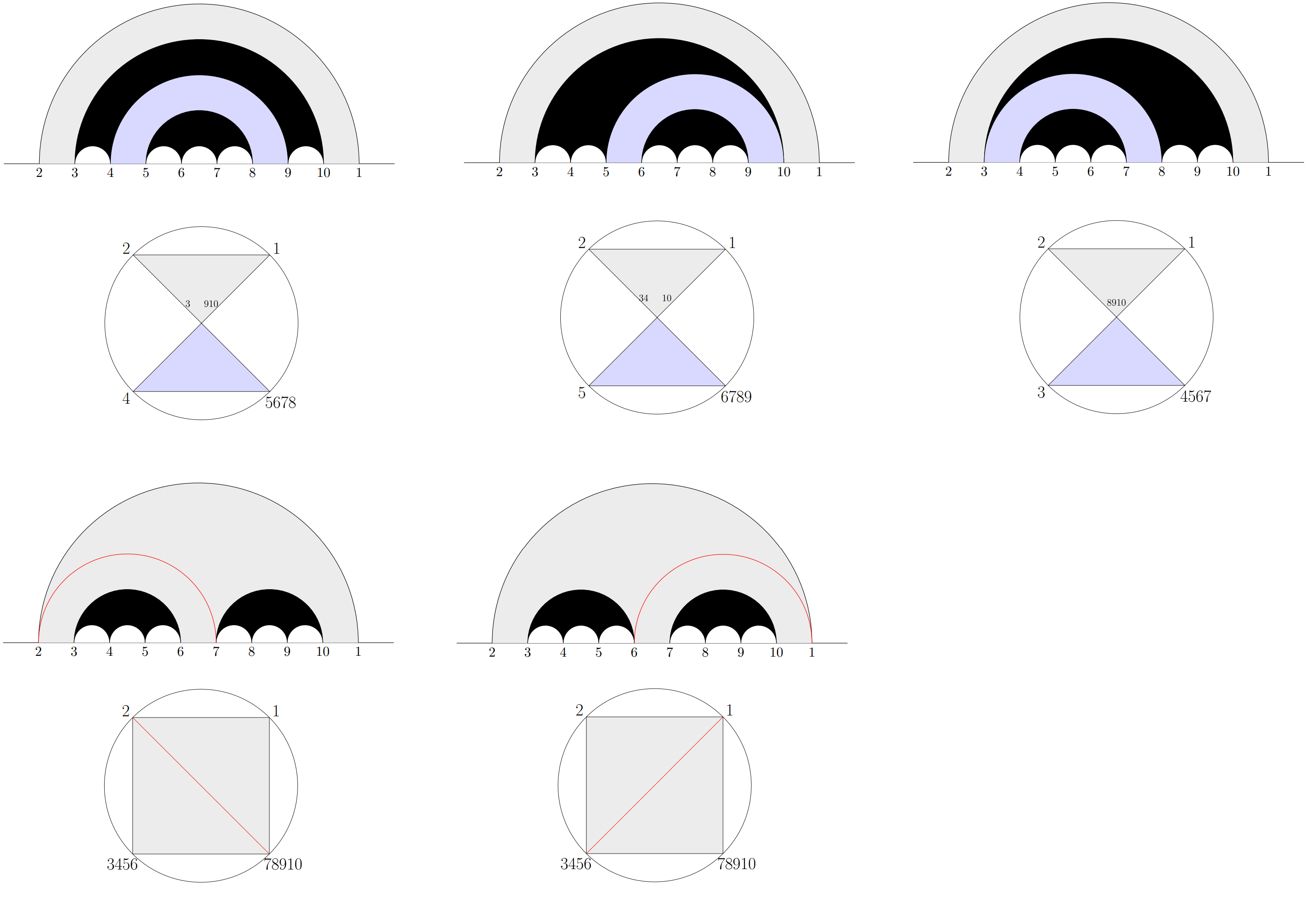}
\centering
\caption{Triangulated extended for $p=6$ and $n=10$ and their corresponding polygon decomposition shown below.}
\label{trinp6n10pol}
\end{figure}
Now let us have a look at figure \ref{trinp5n11pol1}, which belongs to the $p=5$, $n=11$ case. This is the first example in which we encounter with a case where we have a single 3-chord separating three meadows (last diagram in the figure). In this case, we have a triangle with two labelled internal vertices, one of them being label 2, which is a label from an upper chord in a meadow. This is reminiscent of the $\phi^4$ example for $n=16$ (see figure \ref{p4n16}), in which a 5-point meadow was separating three different meadows, and which corresponded to a pentagon with two labelled internal vertices.

Figures \ref{trinp5n11pol2} and \ref{trinp5n11pol3} show the rest of the polygon decompositions for the remaining triangulated extended diagrams in $p=5$ and $n=11$.

\begin{figure}[H]
\includegraphics[width=15cm]{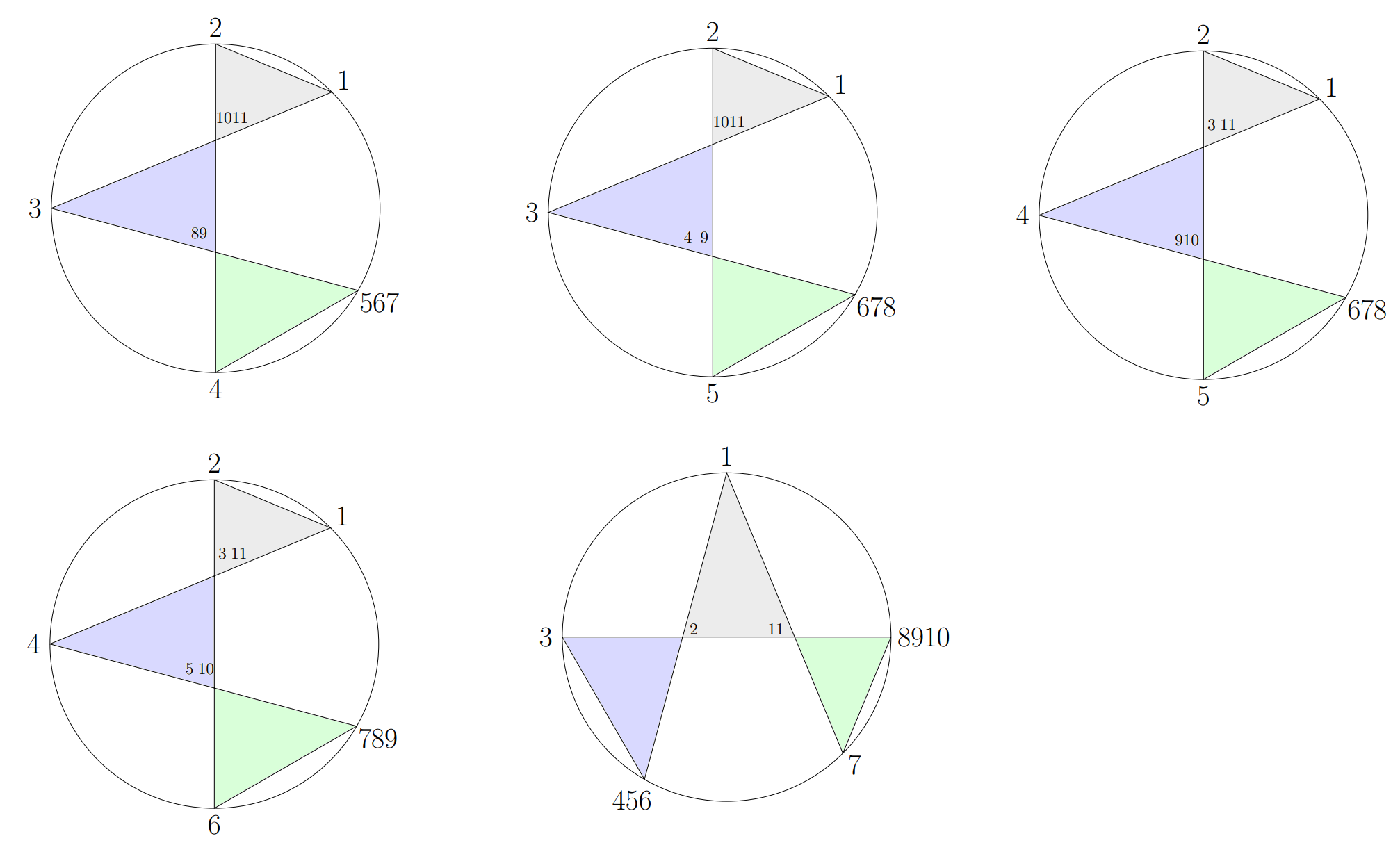}
\centering
\caption{Polygon decomposition of the extended diagrams for $p=5$ and $n=11$ appearing in figure \ref{p5n111}, presented in the same order.}
\label{trinp5n11pol1}
\end{figure}

\begin{figure}[H]
\includegraphics[width=15cm]{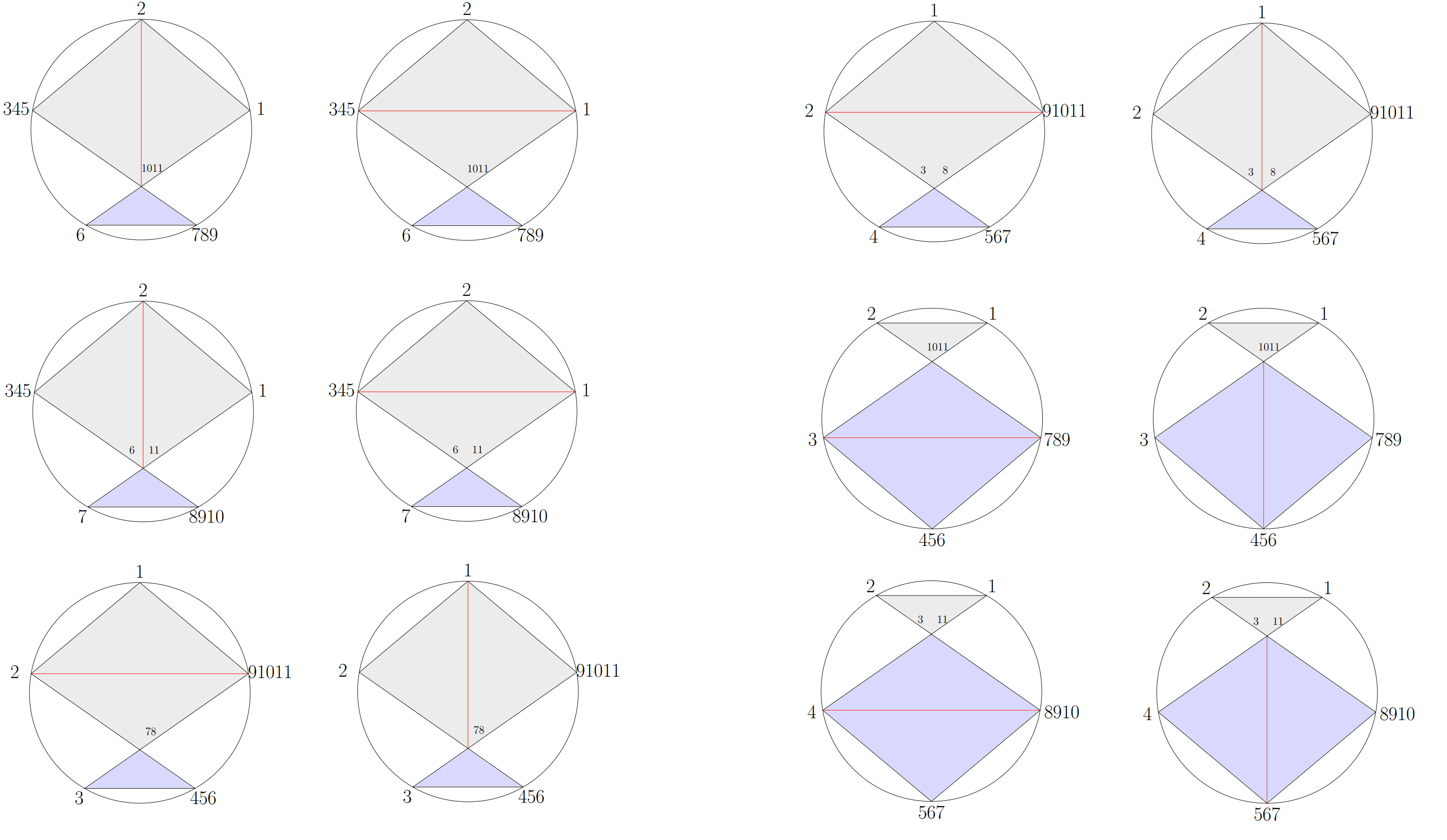}
\centering
\caption{Triangulations of the polygon decompositions for the extended diagrams in $p=5$ and $n=11$ appearing in figure \ref{p5n112}. The first two columns correspond to the first three rows in figure \ref{p5n112}, while the last two columns correspond to the last three rows in figure \ref{p5n112}, ordered in the same way.}
\label{trinp5n11pol2}
\end{figure}

\begin{figure}[H]
\includegraphics[width=15cm]{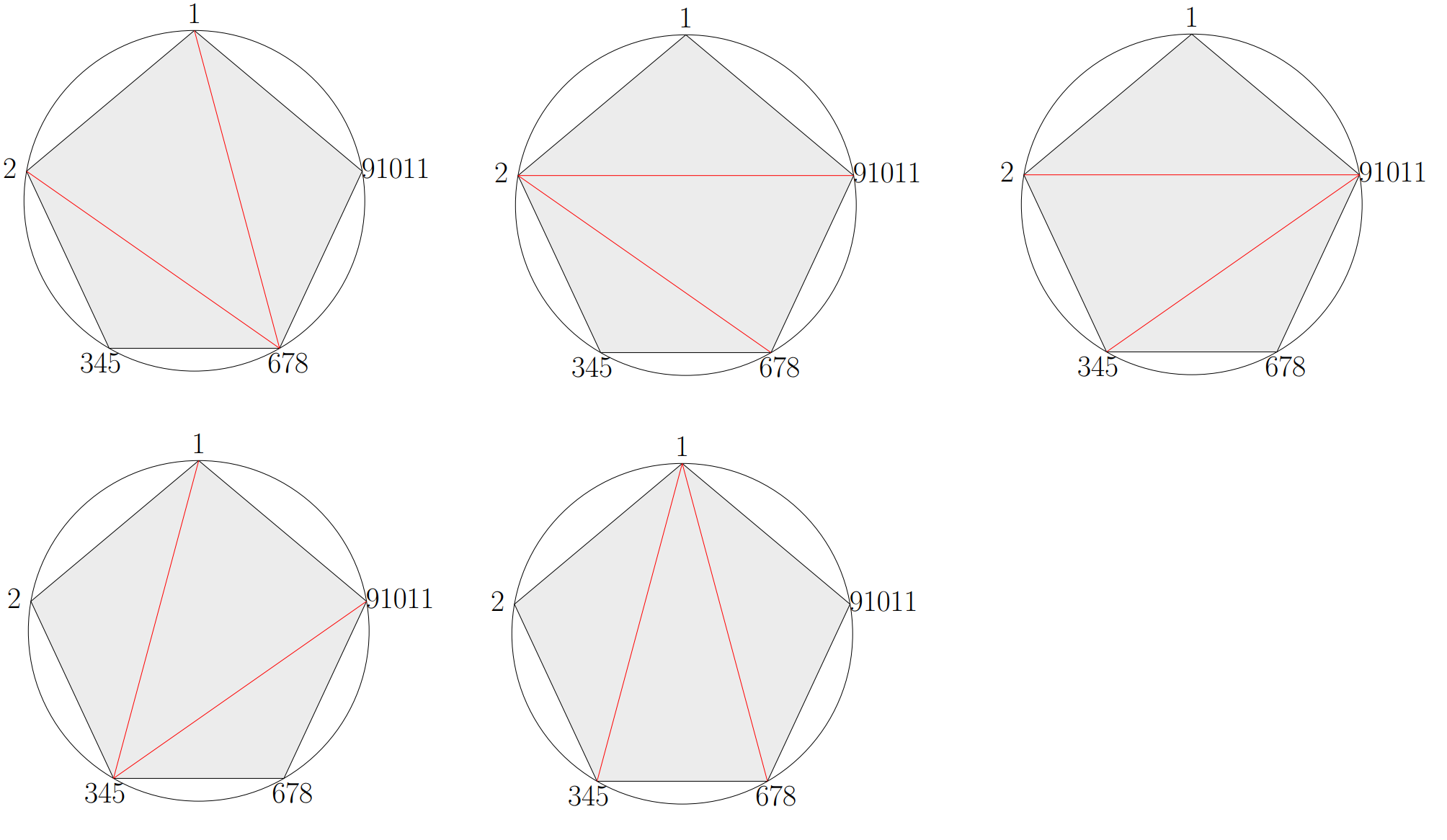}
\centering
\caption{Polygon decomposition of the extended diagram for $p=5$ and $n=11$ appearing in figure \ref{p5n113}, presented in the same order.}
\label{trinp5n11pol3}
\end{figure}

Figures \ref{p5n14} and \ref{p5n14v2} correspond to two similar decompositions for $p=5$ and $n=14$.

\begin{figure}[H]
\includegraphics[width=15cm]{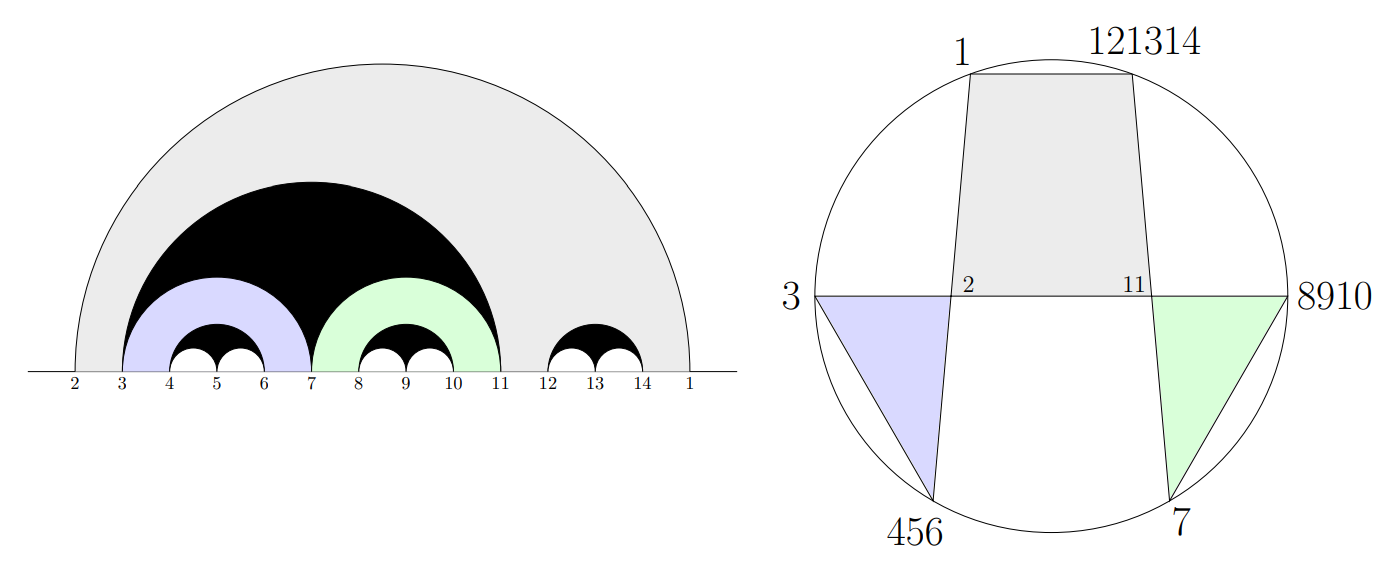}
\centering
\caption{One of the extended diagrams for $p=5$ and $n=14$, together with its polygon decomposition. Note that this is the same extended diagram as the last one in figure \ref{p5n111}, which corresponds to $p=5$ and $n=11$, but with an additional 3-chord joining the new points 12, 13 and 14. In the polygon decomposition representation this is equivalent to adding a new vertex $\overline{12,13,14}$ to the upper triangle in figure \ref{trinp5n11pol1}.}
\label{p5n14}
\end{figure}

\begin{figure}[H]
\includegraphics[width=15cm]{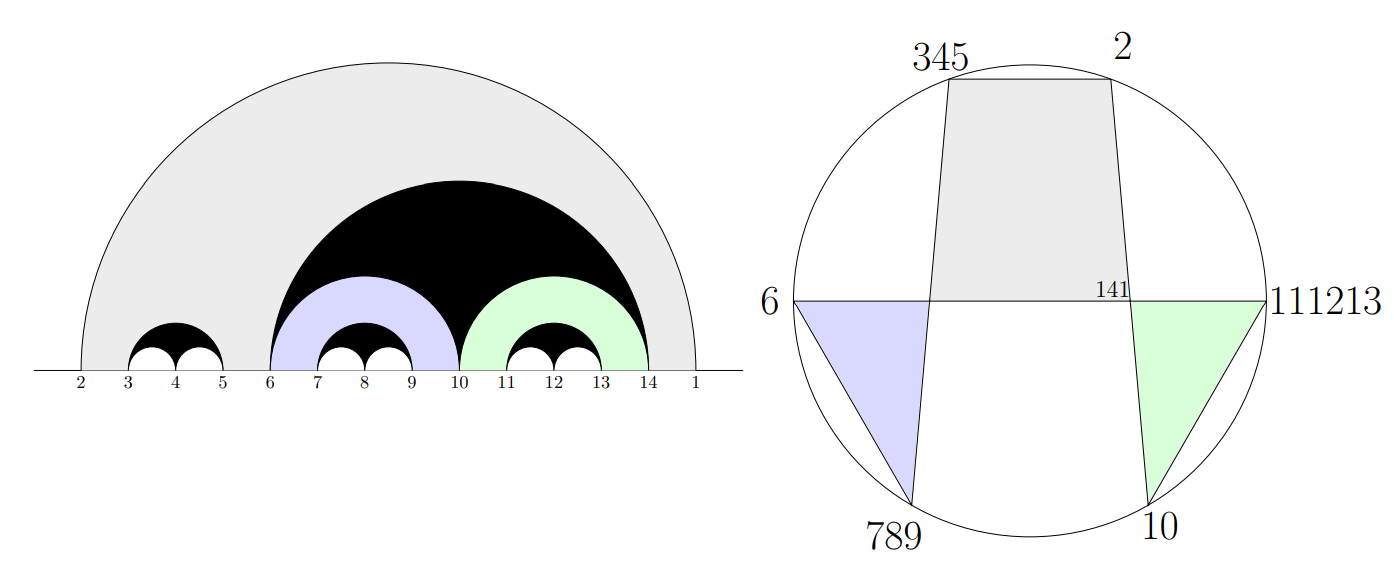}
\centering
\caption{Another similar example of an extended diagram for $p=5$ and $n=14$, together with its polygon decomposition.}
\label{p5n14v2}
\end{figure}

Finally, in figure \ref{p5n14v3} we show a more exotic polygon decomposition of an extended diagram in $p=6$ and $n=14$. Interestingly, now one of the triangles has three internal labels, but one of them is not naturally associated to a vertex.

\begin{figure}[H]
\includegraphics[width=15cm]{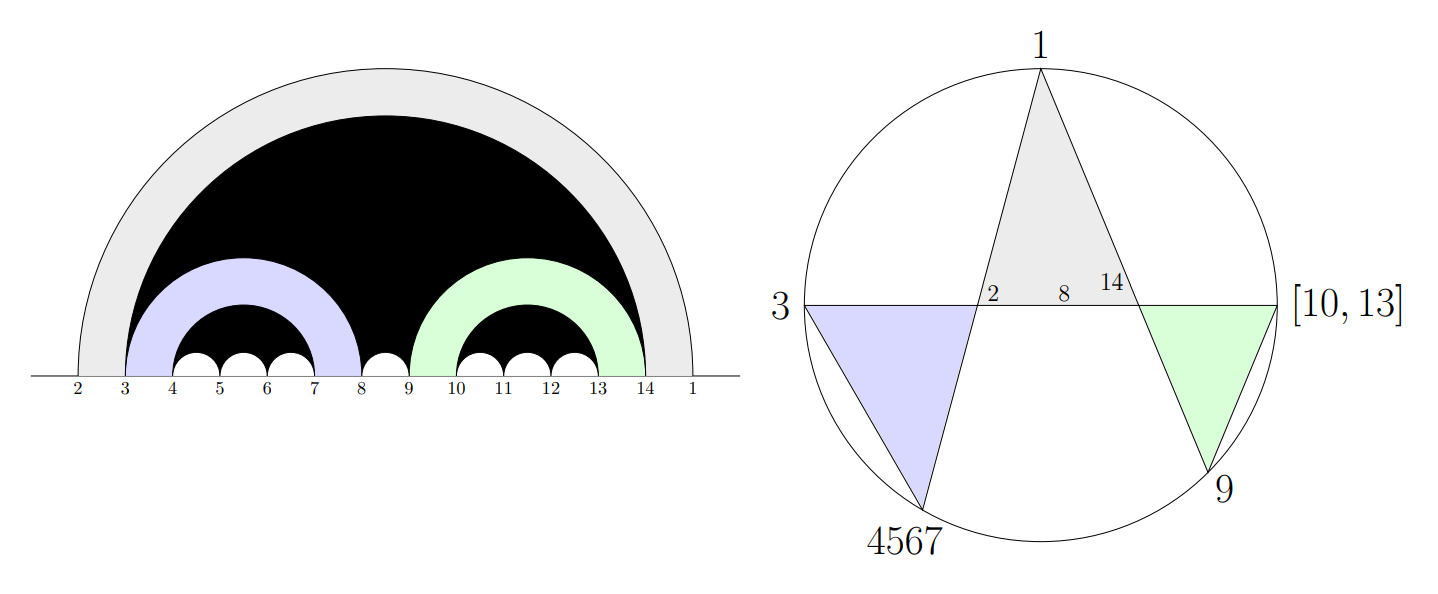}
\centering
\caption{An example of an extended diagram for $p=6$ and $n=14$, together with its polygon decomposition. Adding another 4-chord joining points 15, 16, 17 and 18 in the $n=18$ case adds an extra vertex to the upper triangle, forming a square. The triangulation of this square is not affected by the internal label 8 in the middle.}
\label{p5n14v3}
\end{figure}
From these examples we conclude that a deeper understanding of the relation between extended diagrams and off-diagonal amplitudes is needed. We leave the study of this important and intriguing connection to future research.

\subsection{Towards Smoothly Splitting $\phi^p$ Amplitudes}

In \cite{Cachazo:2021wsz} a new intriguing property of tree-level scattering amplitudes was found. This property was called a 3-split, which comes into play in certain subspaces of the kinematic space where the scattering amplitude smoothly splits into three amputated currents \cite{Dixon:1996wi,Mafra:2016ltu} without becoming singular. This and related phenomena, including factorization near zeros in ``stringy'' amplitudes \cite{Arkani-Hamed:2023swr}, or extensions of the splitting behavior \cite{Cao:2024gln}, have recently been studied.

In this subsection we comment on the possibility of finding subspaces in the kinematic space that would lead to similar smooth splits in $\phi^p$ amplitudes.

For example, it is reasonable to start with $\phi^4$ and $n=10$, since extended diagrams are related to 6-point cubic amplitudes, which is the minimum number of particles required for the smooth split to happen in $\phi^3$ amplitudes. In this case, the extended diagram corresponding to a single 6-point meadow is a diagonal amplitude $m_{{6}}(\mathbb{I},\mathbb{I})$ under the bijection
$$\{X_{1,4}\to s_{{12}},X_{2,5}\to s_{{23}},X_{4,7}\to s_{{34}}, X_{6,9}\to s_{{45}},X_{1,8}\to s_{{56}},X_{2,9}\to s_{{61}},X_{1,6}\to X_{{1},{4}},$$
$$X_{2,7}\to X_{{2},{5}},X_{4,9}\to X_{{3},{6}}\}\,.$$
The conditions $X_{1,6}=X_{1,4}+X_{2,5}$, $X_{2,7}=X_{1,8}+X_{2,9}$ and $X_{4,9}=X_{4,7}+X_{6,9}$ are equivalent, under the bijection, to the split kinematics subspace $({2},{4},{6})$ on 6 particles, given by the conditions $X_{{1},{4}}=s_{{12}}+s_{{23}}$, $X_{{3},{6}}=s_{{34}}+s_{{45}}$ and $X_{{2},{5}}=s_{{56}}+s_{{61}}$. The contribution of the extended diagram in this subspace smoothly splits into
$$\left(\frac{1}{X_{2,5}}+\frac{1}{X_{4,7}}\right)\left(\frac{1}{X_{6,9}}+\frac{1}{X_{1,8}}\right)\left(\frac{1}{X_{1,4}}+\frac{1}{X_{2,9}}\right)\,.$$
It would be interesting to study if the whole $\phi^4$ amplitude can therefore smoothly split. At first it might seem impossible since other extended diagrams are given by products of 4-point or even 3-point and 5-point subamplitudes, but the whole combination might allow for smooth splits, maybe using a different kinematic subspace.

One interesting possibility would be to try and understand how the smooth splitting works at the level of extended diagrams. Of course, by directly looking at Feynman diagrams it is non-obvious what is the combinatorial picture that underlies smooth splits. But, similarly to what is seen in other contexts like the CHY formula \cite{Cachazo:2021wsz} or from the kinematic mesh \cite{Arkani-Hamed:2023swr}, extended diagrams in some sense resum Feynman diagrams, so it might be possible to see this behavior in a more obvious way. We leave this fascinating question for future research.

\section*{Acknowledgements}
We would like to thank F. Cachazo for reading a previous version of this draft, as well as for helpful comments and discussions. We also thank J. Drummond and \"O. G\"urdo\u{g}an for discussions about the topic of this work, and C. Chowdhury for explaining how to use TikZ, which helped to generate the figures of the paper. BGU is supported by the STFC consolidated grant ST/X000583/1. KY is supported by an NSERC Discovery grant, the Canada Research Chairs program, and was supported by the Emmy Noether Fellows program at Perimeter Institute which is supported in part by the Simons Foundation. Research at Perimeter Institute is supported in part by the Government of Canada through the Department of Innovation, Science and Economic Development Canada and by the Province of Ontario through the Ministry of Colleges and Universities.



\bibliographystyle{JHEP}
\bibliography{references}

\end{document}